\documentclass[a4paper,10pt]{article}
\usepackage[utf8]{inputenc}
\usepackage[T1]{fontenc}
\usepackage{comment}

\usepackage{a4wide}
\usepackage{indentfirst}

\usepackage{amsmath,amssymb}
\usepackage{amsthm}

\usepackage{authblk}

\usepackage[dvipsnames,table,xcdraw,svgnames]{xcolor}

\usepackage{hyperref}
\hypersetup{
    breaklinks=true,
    colorlinks=true,
    linkcolor=Navy,
    citecolor=Navy,
    urlcolor=Navy
}

\usepackage{graphicx}
\graphicspath{{./fig/},{./fig-thesis/}}
\usepackage[verbose]{epstopdf}
\usepackage{booktabs}
\usepackage{multirow}

\usepackage{paralist} %
\usepackage{enumitem}

\usepackage[absolute,overlay,showboxes]{textpos}
\TPMargin{2mm}
\textblockrulecolour{Navy}

\usepackage{natbib}

\usepackage{soulutf8}

\soulregister\cite7%
\soulregister\citep7%
\soulregister\eqref7%
\soulregister\pageref7%
\soulregister\ref7%

\usepackage{xifthen}
\newboolean{default}
\newcommand{\case}{}
\newcommand{\default}{}
\newenvironment{switch}[1]{%
    \setboolean{default}{true}
    \renewcommand{\case}[2]{\ifthenelse{\equal{#1}{##1}}{%
        \setboolean{default}{false}##2}{}}%
    \renewcommand{\default}[1]{\ifthenelse{\boolean{default}}{##1}{}}
}{}

\newtheorem{theorem}{Theorem}[section] 
\newtheorem{definition}[theorem]{Definition} 

\newtheorem{lemma}[theorem]{Lemma} 
\newtheorem{proposition}[theorem]{Proposition} 
\newtheorem{remark}[theorem]{Remark} 

\newcommand{\doi}[1]{DOI~\href{\detokenize{http://dx.doi.org/#1}}{\detokenize{#1}}}
\newcommand{\zblnumber}[1]{Zbl~\href{\detokenize{https://zbmath.org/?q=an:#1}}{\detokenize{#1}}}
\newcommand{\mrnumber}[1]{\href{\detokenize{https://www.ams.org/mathscinet-getitem?mr=#1}}{\detokenize{MR#1}}}

\newcommand{\msc}[1]{\href{https://mathscinet.ams.org/mathscinet/msc/msc2020.html?s=#1}{#1}}
\newcommand{\acm}[1]{\href{https://www.acm.org/publications/computing-classification-system/1998/#1}{\uppercase{#1}}}
\newcommand{\jel}[1]{\href{https://www.aeaweb.org/jel/guide/jel.php}{#1}}%

\newcommand{\Dom}{\operatorname{Dom}}
\renewcommand{\d}{\,\mathrm{d}}

\newcommand{\e}{\mathrm{e}}
\newcommand{\D}{\mathbb{D}}
\newcommand{\E}{\mathbb{E}}
\newcommand{\F}{\mathcal{F}}

\renewcommand{\P}{\mathbb{P}}
\newcommand{\Q}{\mathbb{Q}}
\newcommand{\R}{\mathbb{R}}

\newcommand{\1}{\mathbf{1}}

\newcommand{\ccode}[2]{\par
        \vspace*{8pt}
        {{\leftskip18pt\rightskip\leftskip
        \noindent{\it #1}\/: #2\par}}\par}
\newcommand{\keywords}[1]{\ccode{Keywords}{#1}}
\newcommand{\email}[1]{\href{mailto:#1}{#1}}
\def\received#1{Received~#1\par}

\newcommand{\jpTitle}{Computation of Greeks under rough Volterra stochastic volatility models using the Malliavin calculus approach}
\newcommand{\jpAuthors}{M. Al-Foraih, \`{O}. Bur\'{e}s, J. Posp\'{\i}\v{s}il and J. Vives}
\newcommand{\jpKeywords}{Stochastic Volatility; Rough Volatility; Greeks; Malliavin calculus}
\newcommand{\jpMSC}{\msc{60G22}; \msc{91G20}; \msc{91G60}}
\newcommand{\jpACM}{\acm{g.3}}
\newcommand{\jpJEL}{\jel{C58}; \jel{G12}; \jel{C63}}%
\newcommand{\jpDateReceived}{4 July 2025} %
\newcommand{\jpDate}{}%

\author[1]{Mishari Al-Foraih}%
\author[2]{\`{O}scar Bur\'{e}s}%
\author[3]{Jan Posp\'{\i}\v{s}il\thanks{Corresponding author, \email{honik@kma.zcu.cz}}} %
\author[2]{Josep Vives} %
\affil[1]{Department of Mathematics, Faculty of Science, Kuwait University, Al-Shidadiya, Kuwait,\vspace*{3pt}}
\affil[2]{Facultat d'Economia i Empresa, Universitat de Barcelona, \authorcr Diagonal 690--696, 08034 Barcelona, Spain (Catalunya),\vspace*{3pt}}
\affil[3]{NTIS - New Technologies for the Information Society, Faculty of Applied Sciences, \authorcr University of West Bohemia, Univerzitn\'{\i} 2732/8, 301 00 Plze\v{n}, Czech Republic,\vspace*{3pt}}

\ifpdf
\hypersetup{
  pdftitle={\jpTitle},
  pdfauthor={\jpAuthors},
  pdfkeywords={\jpKeywords},
  pdfinfo={
      MSC={\jpMSC},
      JEL={\jpJEL}
  }
}
\fi

\title{\textcolor{Navy}{\textsc{\jpTitle}}}
\date{\jpDate}

\begin{document}

\maketitle

\begin{center}
\received{\jpDateReceived}
\end{center}

\begin{abstract}
Using Malliavin calculus techniques, we obtain formulas for computing Greeks under different rough Volterra stochastic volatility models. Due to the fact that underlying prices are not always square integrable, we extend the classical integration by parts formula to integrable but not necessarily square integrable functionals. First of all, we obtain formulas for general stochastic volatility (SV) models, concretely the Greeks Delta, Gamma, Rho, Vega and we introduce the Greek with respect to the roughness parameter. Then, the particular case of rough Volterra SV models is analyzed. Finally, three examples are treated in detail: the family of alpha-RFSV models, that includes rough versions of SABR and Bergomi models, a mixed alpha-RFSV model with two different Hurst parameters representing short (roughness) and long memory, and the rough Stein-Stein model. For different models and Greeks we show a numerical convergence of our formulas in Monte Carlo simulations and depict for example a dependence of the Greeks on the roughness parameter.
 \end{abstract}

\ccode{MSC classification}{\jpMSC}
\ccode{ACM classification}{\jpACM}
\ccode{JEL classification}{\jpJEL}

\setcounter{tocdepth}{2}
\tableofcontents
\section{Introduction}\label{sec:introduction}

Sensitivity measures or Greeks are essential tools in hedging financial derivatives. Given a financial derivative or a portfolio with derivatives, Greeks measure the impact on its price as a result of changes in the different parameters on which it depends. In the simple case of \cite{BlackScholes73} model, for plain vanilla options, closed-form formulas can easily be obtained. But once one departs from this simple case, that is, for more complex derivatives or for more complex models for the underlying securities, no closed-form formulas are available, and the use of Monte Carlo simulations is needed. Due to the fact that numerical computation of Greeks involves the numerical computation of derivatives, procedures are complicated and slow. 

The application of Malliavin calculus techniques, concretely, the so-called, integration by parts formula, see \cite{Fournie99} and \cite{Fournie01}, improved dramatically the efficiency in terms of computational time required for the numerical computation of Greeks, specially for discontinuous payoffs. For surveys of the Malliavin Calculus technique under the Black-Scholes framework see \cite{Montero03} or Chapter 6 in \cite{Nualart06}.

It is very well-known since thirty years ago that one of the main issues of Black-Scholes model, the assumption of constant volatility of the underlying asset, does not describe correctly the empirical reality observed in financial markets. In fact, realized volatility time series tends to cluster depending on the spot asset level and it certainly does not take on a constant value within a reasonable time-frame, see for example \cite{Cont01}. To deal with such inconsistencies, stochastic volatility (SV) models were proposed originally by \citet{HullWhite87} and later further developed in many other papers becoming a research field itself. Three of the most known SV models are \cite{SteinStein91}, \citet{Heston93} and SABR (see \cite{Hagan02}). In SV models it is assumed that the instantaneous volatility of asset returns is of random nature. Specifically, the latter two approaches, SABR and Heston models, became popular in the eyes of both practitioners and academics. All of cited SV models assume a volatility essentially described by a stochastic differential equation driven by a Brownian motion. 

Different extensions of SV models have been developed. One that is not treated in the present paper is the addition of jumps in the price dynamics or the substitution of the Brownian motions that drives price and volatility by Lévy processes. 

The extension in which we are interested in this paper is based on the fact that independent increments of the Brownian motion turned out to be a limitation in describing the real implied volatilities observed in financial markets. This helped to increase the popularity of the fractional Brownian motion (fBm), a generalization of the Brownian motion that allows the correlation of increments depending on the so-called Hurst index $H \in (0,1)$. If $H=1/2$, one gets the standard Brownian motion, if $H>1/2$, the increments are positively correlated and trajectories are more regular, and if $H<1/2$, the increments are negatively correlated and trajectories are less regular. In the latter case we speak about the rough regime. Existence and uniqueness of a strong solution for stochastic differential equations driven by fractional Brownian motion was proved by \cite{NualartOuknine2002}.

The use of rough volatility models, that is, stochastic volatility models with the volatility dynamics described partially or completely by a fractional Brownian motion or a similar fractional process has since ten years ago also become a research field itself. The first rough volatility model was introduced in \cite{Alos07} in an effort to better describe the smirk of the implied volatility observed in markets near expiration. In \cite{Bayer16}  and \cite{Gatheral18} a detailed analysis of rough fractional stochastic volatility (RFSV) models is given. They seem to be very consistent with market option prices and realized volatility time series, and moreover, they provide superior volatility prediction results to several other models, see \cite{Bennedsen17}. Several approaches to the exact and approximate option pricing formulas in models where volatility is a fractional Ornstein--Uhlenbeck (fOU) or fractional Cox--Ingersoll--Ross (fCIR) process were introduced in \cite{Mishura2019}. A comprehensive survey of continuous stochastic volatility models including the fractional and rough models was written only recently by \cite{DiNunno2023From}.

Computation of Greeks for jump-diffusion models, including local volatility models, is a well-developed topic, see, for example, \cite{Petrou08}, \cite{Eddahbi15}, \cite{Eberlein16}. Two main methods are used: Fourier transform method and Malliavin calculus method. 

But for stochastic volatility models things are less developed. For Heston and Bates models we find formulas, using the Malliavin calculus approach in \cite{Davis06} and in \cite{Mhlanga15}. General formulas to compute Greeks for options under an SV model were obtained by  \cite{El-Khatib09} using Malliavin calculus techniques and in \cite{Khedher12} using both Malliavin calculus techniques and the Fourier transform method. Greeks for SABR model have been obtained in \cite{Yamada17}. \cite{Yilmaz18} calculated Greeks for SV models with both stochastic interest rate and stochastic volatility using Monte Carlo simulation. Finally, \cite{Yolcu18} applied Malliavin calculus for a general SV models and calculated the Delta of different SV models such as Heston and Stein and Stein models. 

The aim of the present paper is to apply Malliavin calculus techniques to compute several Greeks for derivatives under a rough Volterra SV model.

The structure of the paper is the following. In Section~\ref{sec:preliminaries} we recall the definition of Greeks and Malliavin calculus fundamentals and introduce Volterra processes. In Section~\ref{sec:methodology} we introduce rough Volterra models and we derive the Malliavin weights for general Volterra stochastic volatility models. In Section \ref{sec:results1} we derive the formulas for particular models, namely for the $\alpha$RFSV model, introduced for the first time by \cite{MerinoPospisilSobotkaSottinenVives21ijtaf}, for a new mixed $\alpha$RFSV model, and for the rough Stein-Stein model. Other models such as the rough Bergomi model, will be handled as a special case of the $\alpha$RFSV model. We will present also some numerical results, especially we will demonstrate numerically the convergence of selected obtained formulas. We conclude in Section~\ref{sec:conclusion}.

\section{Preliminaries and notation}\label{sec:preliminaries}

In this section we introduce the concept of Greeks (Section \ref{ssec:Greeks}) and the Malliavin calculus fundamentals (Section \ref{ssec:Malliavin}), including the corresponding integration by parts formula. We also give a short introduction to rough Volterra processes (Section \ref{ssec:rVolterra_processes}).

\subsection{Greeks}\label{ssec:Greeks}

Consider a price process $S:=\{S_t, t\geq 0\}$. Let $\{{\cal F}_t, t\geq 0\}$ be the completed natural filtration generated by process $S.$ A financial derivative can be seen as a contingent claim with payoff $F$, where $F$ is a random variable adapted to ${\cal F}_T$ for a fixed expiry date $T>0.$ Assume a fixed instantaneous interest rate $r\geq 0.$ Under the no arbitrage principle, the price at $t\in [0,T]$ of the derivative $F$ is given by

$$P_t:={\E}_{\Q}[\e^{-r(T-t)}F|{\cal F}_t],$$ 
where $\Q$ is a risk-neutral probability measure. 

The price $P_t$ depends on the different parameters of the model that describes the underlying price $S$ like the initial price $S_0$ or the initial volatility or variance $V_0$, the parameters that describe the derivative like the maturity date $T$ or the strike price $K$ in the case of options, and parameters of the market like the interest rate $r.$ 

Greeks or price sensitivities are used in hedging and for measuring and managing risk. They allow to predict the immediate future movements of a derivative price. For simplicity we assume in all the paper that $t=0$ and $F=f(S_T)$ for a certain function $f.$ The Greeks involved in the present paper are given by the following definition.

\begin{definition}[Greeks]\label{d:greeks}
	
Greeks Delta, Gamma, Rho, Vega and the derivative with respect to the Hurst parameter in rough volatility models are defined respectively as

\begin{align}
\Delta &= \partial_{S_0} \E_{\Q}\left[ \e^{-rT} f(S_T)\right], \\
\Gamma &= \partial_{S_0} \Delta=\partial^2_{S_0, S_0}  \E_{\Q}\left[ \e^{-rT} f(S_T)\right],\\
\varrho &= \partial_{r} \E_{\Q}\left[ \e^{-rT} f(S_T)\right], \\
\mathcal{V} &= \partial_{V_0} \E_{\Q}\left[ \e^{-rT} f(S_T)\right], \\
\mathcal{H} &= \partial_{H} \E_{\Q}\left[ \e^{-rT} f(S_T)\right],
\end{align}
where $S_0$ is the initial price of the stock, $V_0$ denotes the initial volatility, $r$ the fixed interest rate and $H$ the Hurst parameter associated to the fractional process underlying the rough Volterra process associated to the volatility. 
\end{definition}

\subsection{Malliavin calculus and integration by parts formula}\label{ssec:Malliavin}

The main Malliavin calculus tool to compute Greeks is the so called integration by parts formula (IBP). We recall in this subsection the IBP formula necessary for our purposes and the basic elements of Malliavin calculus needed to understand it. We refer the reader to the excellent references \cite{Nualart06} and \cite{Nualart18} for all proofs of this subsection and details. 

Let $W:=\{W_t, t\in [0,T]\}$ be a standard Brownian motion defined on a complete probability space $(\Omega, {\cal F}, {\P})$. Recall that it can be seen as an isonormal Gaussian process defined on the Hilbert space $H:=L^2[0,T]$ and write $W=\{W(h), h\in H\}$ where 

$$W(h):=\int_0^T h(s)\d W_s$$
is the Wiener-Itô integral of function $h.$

Malliavin calculus, sometimes also referred as Malliavin-Skorohod calculus, is based on two dual operators, the Malliavin derivative $D$ and the Skorohod integral $\delta$. 

Consider the set $\cal S$ of smooth random variables $F=f(W(h_1), \dots, W(h_n))$ where 
$f\in C^{\infty}_b ({\R}^n)$, the space of infinitely differentiable functions such that they and all its partial derivatives are bounded, and $h_1,\dots h_n$ are elements of $H.$ 

For any $F\in {\cal S}$, its Malliavin derivative is the $H-$valued random variable defined as 

$$DF:=\sum_{i=1}^n (\partial_i f)(W(h_1), \dots, W(h_n))h_i$$
where $\partial_i$ are for any $i=1,\dots, n$ the partial derivatives of function $f.$ 

Note that in particular, for $s,t_1,\dots, t_n\in [0,T],$

$$D_sF:=\sum_{i=1}^n (\partial_i f)(W_{t_1}, \dots, W_{t_n}){1\!\!1}_{[0,t_i]}(s).$$

This operator can be straightforward iterated and we can consider $k-$order derivatives $D^kF$ as elements of $L^p ([0,T]^k).$

Being $\cal S$ dense in $L^1$, it can be shown that derivatives $D^k$ are closed and densely defined operators from $L^p(\Omega)$ to $L^p(\Omega\times [0,T]^k)$, for any $p,k\geq 1,$ with domain ${\D}^{k,p}$ defined as the completion of $\cal S$ by the seminorm 

$$||F||^p_{k,p}:={\E}\left[|F|^p\right] + \sum_{j=1}^k {\E}\left[ ||D^j F||^p_{L^2([0,T]^k)}\right].$$

We can also define, for $k \geq 1$, the space $\D^{k,\infty} = \bigcap_{p \geq 1} \D^{k,p}$. 

Malliavin derivative satisfies a chain rule in the sense that if $\psi:{\R}^n\longrightarrow {\R}$ belongs to $C^1({\R}^n)$ and $F=(F^1,\dots, F^n)$ with $F^i\in {\D}^{1,p}$ for a certain $p\geq 1$ and for any $i=1,\dots, n$, we have 

$$D(\psi(F))=\sum_{i=1}^n (\partial_i \psi)(F)DF^i.$$

The adjoint of the Malliavin derivative operator $D$ is the so called divergence operator $\delta.$ It is an unbounded and closed operator from  $L^2(\Omega\times [0,T])$ to $L^2(\Omega)$, densely defined on a domain $\Dom(\delta)$, in such a way that for any $h \in\Dom(\delta)$ and $F \in {\D}^{1,2}$ we have the duality relationship 

\begin{equation}\label{malliavinduality}
\E[F \delta (h)]= \E \left[ \int_0^T D_tF h_t \d t\right].
\end{equation}

This operator coincides with the so-called Skorohod integral that extends the Itô integral to non-adapted processes in the sense that both integrals coincide for adapted processes $h$ of $L^2(\Omega\times [0,T])$. In this case we write 

\begin{equation}\label{e:delta}
\delta(h)=\int_0^T h_t \d W_t,
\end{equation}
which also implies that
\begin{equation}\label{e:delta1}
\delta(\1_{[0,T]})=W_T.
\end{equation}
We now give the definitions of two vector spaces that are contained in Dom($\delta$).
\begin{definition}
    We define $\mathbb{L}^{1,2}$ as the space of processes $u \in L^2(\Omega \times [0,T])$ such that $u_t \in \D^{1,2}$ for all $t \in [0,T]$.
\end{definition}
In a similar fashion as with $\mathbb{L}^{1,2}$, we define the following space.
\begin{definition}
    We define $\mathbb{L}^{1,\infty}$ as the space of processes $u \in L^{\infty}(\Omega \times [0,T])$ such that $u_t \in \D^{1,\infty}$ for all $t \in [0,T]$.
\end{definition}
\begin{remark}
    One can indeed check that
    \[
    \mathbb{L}^{1,\infty} \subset \mathbb{L}^{1,2} \subset \Dom(\delta).
    \]
    The proof of this inclusions follow from \cite{Nualart06} (see Proposition 1.3.1 and page 42).
\end{remark}
The following two results will be useful for our purposes:

\begin{theorem}\label{t:integral_of_product}
Let $F\in \mathbb{D}^{1,2}$ and $h \in\Dom(\delta).$ Then $Fh\in\Dom(\delta)$ and 
 
\begin{equation}\label{e:compos}
\delta(Fh)=F \delta(h) - \int_0^T D_t F h_t \d t.
\end{equation}
\end{theorem}
\begin{proof}
See \cite{Nualart06} (Proposition 1.3.3)
\end{proof}

\begin{proposition}\label{deriv-int}
Assume $h\in {\mathbb L}^{1,2}$ and the process 
$\{D_t h_s, s\in [0,T]\}$ belongs to $\Dom(\delta)$ for any $t\in [0,T].$ Furthermore, assume there is a version of the process 

$$\left\{\int_0^T D_t h_s \d W_s, t \in [0,T]\right\}$$
in $L^2(\Omega\times [0,T]).$ Then 
$\delta(h) \in \mathbb{D}^{1,2}$ and 
\begin{equation*}
D_t(\delta(h))=h_t+\int_0^T D_t h_s \d W_s
\end{equation*}
\end{proposition}
\begin{proof}
See \cite{Nualart06} (Proposition 1.3.8).
\end{proof}

Finally, we introduce the integration by parts formula and we recall the proof for the sake of completeness. 

\begin{theorem}[IBP formula]\label{t:d_theta}
Let $I$ be an open interval of ${\R}.$ Let $\{F^{\theta}, \theta \in I\}$ and $\{\Phi^{\theta}, \theta \in I\}$ be two families of functionals on ${\D}^{1,2}$, both continuously differentiable with respect to $\theta \in I$. Assume $h \in \mathbb{L}^{1,2}$ such that 

\begin{equation}
\int_0^T D_tF^{\theta} h_t\d t\neq 0 \text{ a.s. on } \{ \partial_{\theta}F^{\theta} \neq 0\} 
\end{equation}
and 

$$\frac{\Phi^{\theta} \partial_{\theta} F^{\theta}\, h}{\int_0^T D_tF^{\theta} h_t\d t}$$ 
belongs to $\Dom(\delta)$ and is continuous in $\theta$. Then, 

\begin{equation}\label{e:d_theta}
\partial_\theta \E[\Phi^{\theta} f(F^{\theta})]
= {\E}\left[ f(F^{\theta})\cdot \delta\left( \frac{\Phi^{\theta} \partial_\theta F^\theta\, h}{\int_0^T D_t F^\theta h(t)\d t}\right) \right] 
+ {\E}[f(F^\theta)\partial_{\theta} \Phi^\theta]
\end{equation}
for any function $f\in C^1_b (\R).$
\end{theorem}
\begin{proof}
The proof follows from \citep[Prop. 6.2.1]{Nualart06} and its extended form from \cite[Prop. 5]{Yolcu18}. From the chain rule we have 

$$f^{\prime}(F^{\theta})=\frac{\int_0^T D_s f(F^\theta) h_s \d s}{\int_0^T D_s F^\theta h_s \d s}.$$ 

Then, using the duality between $D$ and $\delta$ and using the linearity of the expectation we have 

\begin{align*}
\partial_\theta \E[\Phi^\theta f(F^\theta)]
&= {\E}[f^{\prime}(F^{\theta}) \Phi^{\theta}  \partial_{\theta} F^{\theta}] + {\E}[f(F^\theta)\partial_{\theta} \Phi^\theta]  \\
&= {\E}\left[\frac{\int_0^T  D_t f(F^\theta) h_t \d t}{\int_0^T D_t F^\theta h_t \d t}\cdot \Phi^\theta \cdot \partial_{\theta} F^{\theta}\right]+ {\E}[f(F^\theta)\partial_{\theta} \Phi^\theta] \\
&= {\E}\left[f(F^{\theta}) \delta\left(\frac{\partial_{\theta} F^{\theta}\cdot  h\cdot \Phi^{\theta}}{\int_0^T D_t F^\theta h_t \d t}\right)\right] + {\E}[f(F^\theta)\partial_{\theta} \Phi^\theta]
\end{align*}
and this proves the result. 
\end{proof}

This integration by parts formula, however, it is not enough for our purposes. The following proposition is an extension of the previous formula for a more general class of functionals.

\begin{proposition} \label{p: IBP extension}
    Let $I$ be an open interval of ${\R}$. Let $\{ F^{\theta}, \theta \in I\}$ and $\{\Phi^{\theta}, \theta \in I\}$ be two families of functionals on $\D^{1,1}$ and $\D^{1,2}$ respectively, both continuously differentiable with respect to $\theta \in I$. Assume as well that there exists $h \in \mathbb{L}^{1, \infty}$ such that 
    \begin{equation}
    \int_0^T D_tF^{\theta} h_t\d t\neq 0 \text{ a.s. on } \{ \partial_{\theta}F^{\theta} \neq 0\} 
    \end{equation}
    and 

    \begin{equation}
    \frac{\Phi^{\theta}\partial_{\theta}F^{\theta} h}{ \int_0^T D_tF^{\theta} h_t \d t}
    \end{equation}
   belongs to $\Dom(\delta)$ and is continuous in $\theta$. Then,
    \begin{equation}
        \partial_{\theta} \E \left[ \Phi^{\theta}f(F^{\theta})\right] = \E \left[ f(F^{\theta})\cdot \delta \left( \frac{\Phi^{\theta}\partial_{\theta} F^{\theta} h}{\int_0^T D_tF^{\theta} h_t \d t}\right) \right] + \E[f(F^{\theta}) \partial_{\theta}\Phi^{\theta}]
    \end{equation}
    for any function $f \in \mathcal{C}^1_b(\R).$
\end{proposition}

\begin{proof}
    The idea of the proof is to approximate $F^{\theta}$ by a sequence of functionals in $\D^{1,2}$, applying the integration by parts formula (Theorem \ref{t:d_theta}) to the approximating sequence and then conclude the result taking limits. Consider then, for each $\theta \in I$, a sequence $\{F^{\theta}_{n}, n \geq 1\} \subseteq \D^{1,2}$ such that $||F^{\theta}_n - F^{\theta}||_{1,1} \to 0$ as $n \to \infty$. Notice that, by replacing $\{F^{\theta}_n, n \geq 1\}$ by a proper subsequence, we can assume without loss of generality that $F^{\theta}_n$ and $D_tF^{\theta}_n$ converge to $F^{\theta}$ and $D_tF^{\theta}$ almost surely for every $t \in [0,T]$. Using the chain rule for the approximating sequence we have that
    \begin{equation} \label{e: IBP approximated}
         \int_0^T D_t f(F^{\theta}_n)h_t \d t = f'(F_n^{\theta})\int_0^T D_t F^{\theta}_n h_t \d t.
    \end{equation}
    Now from the fact that $h \in \mathbb{L}^{1, \infty}$, $f'$ is bounded and $\E \left[||DF^{\theta}_n - DF^{\theta}||_{L^2[0,T]}\right] \to 0$ as $n \to \infty$, equation \eqref{e: IBP approximated} converges in $L^1(\Omega)$ to
    
    \begin{equation}
       \int_0^T D_tf(F^{\theta})h_t \d t= f'(F^{\theta})\int_0^T D_t F^{\theta} h_t \d t
    \end{equation}
    
    or, equivalently,

    \begin{equation}
         f'(F^{\theta}) = \frac{\int_0^T D_tf(F^{\theta})h_t \d t}{\int_0^T D_t F^{\theta} h_t \d t}.
    \end{equation}
    Now the proof concludes in the same way as in Theorem \ref{t:d_theta}.
\end{proof}

Notice that the previous result is useful provided such a process $h$ exists. In Section 4 we prove that in practical cases such an $h$ always exists and it can be chosen so that $||\1_{[0,T]} - h_t||_{1,1}$ is arbitrarily small.

\subsection{Rough Gaussian Volterra processes}\label{ssec:rVolterra_processes}

A Gaussian Volterra process is defined as a process $Y = \{Y_t, t\geq 0\}$ that can be represented as 

\begin{equation} \label{e:Volterra}
Y_t = \int_0^t K(t,s)\,\d W_s,
\end{equation}
where $W$ is a standard Brownian motion, $K(t,s)$ is a kernel

$$K:[0,T]\times [0,T]\longrightarrow {\mathbb R}$$
such that the following three integrals

$$\int_0^T \int_0^t K(t,s)^2 \d s \d t, \quad \int_0^t K(t,s)^2 \d s, \quad \int_s^T K(t,s)^2 \d t$$
are finite, and for any $t\in [0,T],$

\begin{equation}
\F^Y_t = \F^W_t. \tag{A2}\label{A2}
\end{equation}

By $r(t,s)$ we denote the autocovariance function of $Y_t$ and by $r(t):=r(t,t)$ its variance, that is, 

\begin{align}
r(t,s) &:= \E[Y_t Y_s], \quad t,s\geq 0, \notag \\
r(t) &:= \E[Y_t^2], \quad t\geq 0. \label{e:r}
\end{align}

Different kernels $K$ give different Gaussian Volterra processes. The most famous example is the \emph{standard fractional Brownian motion} (fBm) that is a process with stationary increments and that corresponds to the representation 

\begin{equation}\label{e:fBm}
B_t^H = \int_0^t K_H(t,s)\,\d W_s,
\end{equation}
where $K_H(t,s)$ is a quite complicated kernel that depends on the so-called Hurst parameter $H \in (0,1)$. See for example \cite{Nualart06} (Chapter 5).

The autocovariance function of $B^{H}_{t}$ is given by
\begin{equation}\label{e:fBmCov}
r(t,s):=\E[B^H_t B^H_s] = \frac12 \left( t^{2H} + s^{2H} - |t-s|^{2H}\right), \quad t,s\geq 0,
\end{equation}
and in particular $r(t):=r(t,t) = t^{2H}$, $t\geq0$. For $H=1/2$, fBm is the standard Brownian motion, for $H>1/2$, the increments are positively correlated and the trajectories are more \emph{regular} than Brownian ones, and for $H<1/2$, the increments are negatively correlated and the trajectories are more \emph{rough} than Brownian ones. The term \emph{rough fractional model} therefore refers to the case where $H<1/2$ is considered in the model.

Other kernels $K_H (t,s)$ give the autocovariance function (\ref{e:fBmCov}) despite they loose the property of stationary increments. In the present paper, for modeling purposes, and following many references in the literature, see for example \cite{Alos00}, we consider the so called simplified Riemann-Liouville kernel

\begin{equation}\label{e:K_H}
K_H(t,s) := \sqrt{2H} (t-s)^{H-1/2}.
\end{equation}

Recall that neither the standard fBm nor the simplified Riemann-Liouville process are semimartingales, see for example \cite{Thao06}. This motivates the use of the so-called \emph{approximate fractional Brownian motion} (afBm), i.e. a Gaussian Volterra process with kernel

\begin{equation}\label{e:K_H_eps}
{K}_H(t,s) := \sqrt{2H}(t-s+\varepsilon)^{H-1/2}, \quad \varepsilon\geq 0, \quad H\in(0,1). 
\end{equation}
For every $\varepsilon>0$ such a process is a semimartingale and as $\varepsilon$ tends to zero it converges to the simplified Riemann-Liouville process in the $L^2-$norm, uniformly in $t\in [0,T];$ see \cite{Thao06}. In this case

\begin{align}
r(t,s) &= \int_0^{t\wedge s} {K}_H(t,v) {K}_H(s,v) \d v, \notag \\
r(t)   &= \int_0^t {K}_H^2(t,v) \d v = 2H \int_0^t (t-v+\varepsilon)^{2H-1} \d v = (t+\varepsilon)^{2H} - \varepsilon^{2H}, \notag
\end{align}
Note that if $\varepsilon=0$, we get exactly the variance $r(t)=t^{2H}.$

The following quantities will be useful later: 
$$\kappa_t := \int_0^t {K}_H(t,s)\d s = \sqrt{2H}\int_0^t (t-s+\varepsilon)^{H-1/2} \d s,$$
$$ \frac{\partial}{\partial H} {K}_H(t,s) = {K}_H(t,s)\cdot\left(\frac1{2H}+\ln(t-s+\varepsilon) \right), $$
and
\[
\kappa_t' = \int_0^t \partial_H K_H(t,s) \d s.
\]
Note that, if $H=1/2$ we have $\kappa_t=t.$

To avoid confusion, all theoretical calculations will be provided with a general kernel $K_H(t,s)$ and for numerical purposes, the kernel \eqref{e:K_H_eps} with a suitable value of $\varepsilon$ (typically $10^{-6}$) will be considered.

\section{Methodology for computation of Greeks}\label{sec:methodology} 

In this section we present the main theoretical result, the formulas for Greeks in the general SV model (Section \ref{ssec:General_SV}). Depending on the particular choice of the volatility process, two classes of rough Volterra SV models will be considered (Section \ref{ssec:rVolterra_SV}).

\subsection{General SV model}\label{ssec:General_SV}

We assume a general SV model on a filtered probability space generated by two independent Brownian motions $W$ and $\widetilde W$, under a risk neutral measure. The price process is assumed to follow the equation 

\begin{equation}\label{e:dS_t}
\d S_t = r S_t \d t + \sigma(V_t) S_t \d W_t, \, t\in [0,T],
\end{equation}
where $r\geq 0$ is the constant instantaneous interest rate, $\sigma\in C^2(\mathbb R)$ and $V$ denotes the volatility or the variance process that we assume adapted to the completed filtration generated by $W$ and $\widetilde W.$ Processes $S$ and $V$ are assumed to have continuous trajectories and such that $S_t$ and $V_t$ belong to ${\D}^{2,1}$ and ${\D}^{2,2}$ respectively for any $t\in [0,T].$ Here ${\D}^{2,p}$ is the Sobolev space defined in the previous section associated to the Malliavin derivative with respect process $W$, the driving process of the price process.  Note that $V_t$ can be seen as a functional of $W$ depending on an independent source of randomness $\widetilde W$ that can be treated as deterministic, see \cite{Nualart06} (page 31).

Note that the solution of the price process is given by 

\begin{equation}\label{e:S_t}
S_t=S_0\exp\left\{\int_0^t \left(r-\frac{1}{2} \sigma^2(V_s)\right)\d s+\int_0^t \sigma(V_s)\d W_s\right\}, \, t\in [0,T].
\end{equation}

The following lemmas give the Malliavin derivatives of $S_T$ in terms of the Malliavin derivatives of $V_t.$ All the proofs are quite straightforward using that the Itô integral in \eqref{e:S_t} can be seen as a Skorohod integral and using Proposition~\ref{deriv-int}. Similar computations can be found for example in \cite{El-Khatib09} and \cite{Yolcu18}.

Similarly as in \cite{Yolcu18} (Def. 3, Prop. 6), we have the following result 

\begin{lemma}\label{l:DS}
The Malliavin derivative of $S_T$ with respect to $W$ is 
$$D_tS_T=S_T G(t,T)$$ 
where 
\begin{equation}\label{gformula}
G(t,T)=\sigma(V_t)+\int_t^T \sigma^{\prime}(V_s) D_tV_s \d W_s-\int_t^T \sigma(V_s)\sigma^{\prime}(V_s) D_tV_s \d s.
\end{equation}
\end{lemma}

\begin{proof}
The proof is a straightforward computation of the Malliavin derivative of \eqref{e:S_t} using Proposition~\ref{deriv-int}.
\end{proof}

\begin{lemma}\label{l:DG}
The Malliavin derivative of $G(t,T)$ given above is 
\begin{align*}
D_s G(t,T)
&=\sigma^{\prime}(V_{s\vee t})D_{s\wedge t}V_{s\vee t}\\
&\quad+\int_{s\vee t}^T \sigma^{\prime\prime}(V_u) D_s V_u D_t V_u \d W_u\\
&\quad+\int_{s\vee t}^T \sigma^{\prime}(V_u) D_s D_t V_u \d W_u\\
&\quad-\int_{s\vee t}^T  (\sigma^{\prime}(V_u))^2 D_s V_u D_t V_u \d u\\
&\quad-\int_{s\vee t}^T \sigma^{\prime\prime}(V_u)\sigma(V_u) D_s V_u D_t V_u \d u\\
&\quad-\int_{s\vee t}^T \sigma^{\prime}(V_u)\sigma(V_u) D_s D_t V_u \d u.
\end{align*}
\end{lemma}
\begin{proof}
The proof follows straightforwardly from Proposition~\ref{deriv-int}.
\end{proof}

Note that $V$ is an adapted process and then $D_t V_s$ is null if $t>s$ and $D_s D_t V_u$ is null if $s$ or $t$ is greater than $u$.
\vspace{0.5cm}

Let $\theta$ be a model parameter. Our goal is to calculate the corresponding Greek, in particular the partial derivative with respect to it, see Definition~\ref{d:greeks}. For example from \eqref{e:S_t} and using Fubini's theorem we may write
\begin{equation*}
\partial_{\theta} S_T^\theta = S_T^\theta\cdot\left\{\int_0^T b(V_s)\d s +\int_0^T a(V_s)\d W_s\right\},
\end{equation*}
where we denoted
\begin{equation}\label{e:ab}
a(V_s) = \partial_{\theta} \left(\sigma(V_s)\right)\qquad\text{ and }\qquad b(V_s)= \partial_{\theta} \left(r-\frac12 \sigma^2(V_s) \right).
\end{equation}
Direct consequence of the IBP formula applied to a suitable process $h$ and Lemma \ref{l:DS} then gives us the general Greek formula
\begin{align}
\partial_\theta \E_{\Q}[\e^{-rT} f(S_T^\theta)] 
&= \e^{-rT} \E_{\Q}\left[ f(S_T^\theta)\cdot \delta\left( \frac{\partial_\theta S_T^\theta \cdot h}{\int_0^T S_T^\theta G(u,T) h_u \d u}\right) \right] + \E_{\Q}\left[f(S_T^\theta)\cdot\partial_{\theta}\left(\e^{-rT}\right)\right] \notag \\
&= \e^{-rT} \E_{\Q}\left[ f(S_T^\theta)\cdot \delta(c(V, h)) \right] + \E_{\Q}\left[f(S_T^\theta)\cdot\partial_{\theta}\left(\e^{-rT}\right)\right], \label{e:general_Greek}
\end{align}
where we denoted
\begin{equation}\label{e:c}
c(V, h) = \dfrac{\int_0^T b(V_s)\d s +\int_0^T a(V_s)\d W_s}{\int_0^T G(u,T)h_u\d u}\cdot h.
\end{equation}

In the following theorem we show that the computation of quantities  
\begin{equation}\label{e:iint_DsG}
\int_0^T\int_s^T D_s G(t,T) \d t \d s
\end{equation}
and 
\begin{equation}\label{e:int_G}
\int_0^T G(t,T) \d t
\end{equation}
that depend on the model chosen for the volatility process 
is one of the the essential problems in order to find Malliavin weights and to compute the considered Greeks. More concretely, the two objects depending on the model are $G(t,T)$ and $D_s G(t,T)$ with $s\leq t$.

\begin{theorem}\label{t:greeks_general}
Let $G(t,T)$ be defined as in \eqref{gformula}. Then the formulas for Greeks Delta, Gamma, Rho, Vega and the derivative with respect to $H$ respectively are

\begin{align*}
\Delta 
&= \E_{\Q}\left[ f(S_T)\, \delta\left( \frac{\e^{-rT} \partial_{S_0} S_T\cdot h}{\int_0^T D_u S_T h_u \d u}\right) \right]\\
&= \frac{\e^{-rT}}{S_0}\E_{\Q}\left[f(S_T) \delta \left(\frac{h}{\int_0^T G(u,T) h_u \d u} \right)\right], \\
\Gamma &= \partial_{S_0} \frac{\e^{-rT}}{S_0}\E_{\Q} \left[ f(S_T) \delta \left(\frac{h}{\int_0^T G(u,T) h_u \d u}  \right) \right] \\
&= -\frac{\e^{-rT}}{S_0^2}\E_{\Q} \left[ f(S_T) \delta \left(\frac{h}{\int_0^T G(u,T) h_u \d u} \right) \right] \\
& +  \frac{\e^{-rT}}{S_0^2}\E_{\Q} \left[ f(S_T) \delta \left( \frac{ h}{\int_0^T G(u,T)h_u \d u} \delta \left(\frac{ h}{\int_0^T G(u,T)h_u \d u} \right)\right) \right] \\
\varrho &= \E_{\Q}\left[ f(S_T)\,   \delta\left( \frac{\e^{-rT} \partial_r S_T \cdot h  }{\int_0^T D_u S_T h_u   \d u}\right) \right] 
- \E_{\Q}\left[ f(S_T)\e^{-rT} T \right], \\
&= -T\e^{-rT} \E_{\Q}[f(S_T)] + T\e^{-rT}\E_{\Q}\left[ f(S_T) \delta \left( \frac{h}{\int_0^T G(u,T) h_u \d u})\right)\right], \\
\mathcal{V} &= \E_{\Q} \left[ f(S_T) \delta \left( \frac{\e^{-rT} \partial_{V_0}S_T \cdot h}{\int_0^T D_u S_T h_u \d u} \right) \right] =\e^{-rT} \E_{\Q}[f(S_T) \cdot \delta(c_{V_0}(V,h))], \\
\mathcal{H} &= \E_{\Q}\left[ f(S_T) \delta \left( \frac{\e^{-rT} \partial_{H}S_T \cdot h}{\int_0^T D_u S_T h_u \d u} \right) \right] =\e^{-rT} \E_{\Q}[f(S_T) \cdot \delta(c_{H}(V,h))].
\end{align*}
\end{theorem}
\begin{proof}
    We will show a formula for a general Greek of the form
    \[
    \partial_{\theta} \E_{\Q}[f(F^{\theta}) \Phi^{\theta}].
    \]
    The conclusion of the theorem will follow by replacing $F^{\theta} = S_T$ and $\Phi^{\theta} = \e^{-rT}$. 
    
    Without loss of generality, assume that we are working under the canonical probability space, i.e. $\Omega = \mathcal{C}_0([0,T])$ is the space of continuous functions in $[0,T]$ vanishing at the origin endowed with the supremum norm $||\cdot||_{0}$, thus making $\Omega$ a Banach space. In other words, for $\omega \in \Omega$, $||\omega||_0 = \sup_{t \in [0,T]} |\omega(t)|$. Let $\mathcal{F}$ the Borel $\sigma$-field of $\Omega$ and let $P$ be the probability measure such that the canonical process $W_t(\omega) = \omega(t)$ is a Brownian motion. 
    
    Given two sequences of positive real numbers $\{\delta_n ; n \geq 1\}$, $\{M_n, n \geq 1\}$ such that $\delta_n \to 0$ in a monotonically decreasing way and $M_n \to \infty$ in a monotonically increasing way, we can define the sets $A^{j}_n$ for $j = 1, 2,3, 4, 5$ and $A_n^{\theta}$ as follows:
    \begin{align*}
        &A^{1. \theta}_n = \left \{ \omega \in \Omega;  \sup_{\omega' \in \Omega; ||\omega'||_0 \leq 1/n}\delta_n <  F^{\theta}(\omega + \omega') < M_n \right \}, \\
        &A^{2, \theta}_n = \left \{ \omega \in \Omega;  \sup_{\omega' \in \Omega; ||\omega'||_0 \leq 1/n}\delta_n < |\partial_{\theta} F^{\theta}(\omega + \omega')| < M_n \right \}, \\
        &A^{3,\theta}_n = \left \{ \omega \in \Omega;  \sup_{\omega' \in \Omega; ||\omega'||_0 \leq 1/n}\delta_n < \left| \int_0^T D_t F^{\theta}(\omega + \omega') \mathrm{d} t \right| < M_n \right \}, \\
        &A^{4,\theta}_n = \left \{ \omega \in \Omega;  \sup_{\omega' \in \Omega; ||\omega'||_0 \leq 1/n}\delta_n < \left| \int_0^T D_t \partial_{\theta} F^{\theta}(\omega + \omega') \mathrm{d} t \right| < M_n \right \}, \\
        &A^{5, \theta}_n =  \left \{ \omega \in \Omega;  \sup_{\omega' \in \Omega; ||\omega'||_0 \leq 1/n}\delta_n < \left| D_t\int_0^T D_s F^{\theta}(\omega + \omega') \mathrm{d} s \right| < M_n  \text{ for all } t\in [0,T] \right\}, \\
        & A^{\theta}_n = A^{1, \theta}_n \cap A^{2, \theta}_n \cap A^{3, \theta}_n \cap A^{4, \theta}_n \cap A^{5,\theta}_n.
    \end{align*}
    We also define the random variable
    \[
     \rho_{A^{\theta}_n}(\omega) = \inf \{||\omega'||_0; \omega + \omega' \in A^{\theta}_n \}.
    \]
    Consider the smooth function $\phi(t) \in \mathcal{C}^{\infty}_0(\R)$ satisfying that $\phi(t) = 1$ in $|t| \leq \frac{1}{3}$, $\phi(t) = 0$ if $|t| \geq \frac{2}{3}$ and $|\phi'(t)| \leq 4$ for all $t \in \R$. Then, we define the sequence of random variables 
    \[
    F_n^{\theta} = F^{\theta} \phi(n \rho_{A^{\theta}_n}).
    \]
    The first claim we want to show is that $\partial_{\theta} \phi(n \rho_{A^{\theta}_n}(\omega)) = 0$ almost surely. To see this, fix $\theta_0$ and assume for the sake of simplicity that there exists $\omega' \in \Omega$ such that $||\omega'||_0 = \rho_{A_n^{\theta_0}}(\omega)$ (the general claim follows by replacing the existence of $h$ by the existence of a sequence of elements $\{\omega'_k; k \geq 1\}$ such that $||\omega'_k||_0 \to \rho_{A_n^{\theta_0}}(\omega)$ but the argument is identical). Now, since $||\omega'||_0 = \rho_{A_n^{\theta_0}}(\omega)$, it holds that $\omega + \omega' \in A_n^{\theta_0}$. Now, due to the definition of the set $A_n^{\theta_0}$ and the continuity with respect to $\theta$ of all the random variables involved  we conclude that $\omega + \omega' \in A_n^{\theta}$ for all $\theta$ in a neighborhood of $\theta_0$. Since $\rho_{A_n^{\theta_0}}$ is defined as an infimum we deduce that
    \[
    \rho_{A_n^{\theta_0}}(\omega) \geq \rho_{A_n^{\theta}}(\omega)
    \]
    for all $\theta$ in a neighborhood of $\theta_0$. Since the choice of $\theta_0$ is arbitrary, we conclude that the function $\rho_{A_n^{\theta}}(\omega)$ is locally constant almost surely and therefore $\partial_{\theta} \rho_{A_n^{\theta}}(\omega) = 0$ for every $\omega \in \Omega$.
    
    Then, by the hypothesis on $F^{\theta}$ and following the same argument as in \cite[p. 231]{Nualart06} we have that $A_n^{\theta} \uparrow \Omega$ as $n \to \infty$, $F_n^{\theta} = F^{\theta}$, $\partial_{\theta} F^{\theta}_n =  \partial_{\theta} F^{\theta}$ on $A_n^{\theta}$ and $\partial_{\theta} F^{\theta}_n \in \mathbb{D}^{1,2}$, $F^{\theta}_n \in \mathbb{D}^{2,2}$. Consider now a localized Greek of the form
    \[
    \partial_{\theta} \E_{\Q} \left[ f(F^{\theta}) \Phi^{\theta} \phi(n\rho_{A^{\theta}_n}) \right] = \E_{\Q} \left[ f'(F^{\theta}) \partial_{\theta} F^{\theta} \Phi^{\theta} \phi(n \rho_{A_n^{\theta}})\right] + \E_{\Q} \left[f(F^{\theta}) \partial_{\theta} \Phi^{\theta} \phi(n\rho_{A_n^{\theta}}) \right].
    \]
    Applying the IBP formula in Theorem \ref{t:d_theta} to the first term we obtain that
    \[
    \partial_{\theta} \E_{\Q} \left[ f(F^{\theta}) \Phi^{\theta} \phi(n\rho_{A^{\theta}_n}) \right] = \E_{\Q} \left[ f(F^{\theta}) \delta\left( \frac{\partial_{\theta} F^{\theta}  \Phi^{\theta} \phi(n \rho_{A_n^{\theta}})h}{\int_0^T D_tF^{\theta} h_t \mathrm{d} t}\right)\right] + \E_{\Q} \left[f(F^{\theta}) \partial_{\theta} \Phi^{\theta} \phi(n\rho_{A_n^{\theta}}) \right].
    \]
    Notice that, by definition of $\phi(n\rho_{A_n^{\theta}})$, all the terms in the right-hand-side of the equation are well defined. We now will deduce the general formula for the Greeks using a limiting procedure. On the one hand, since $\phi(n \rho_{A_n^{\theta}}) \to 1$ almost surely and it is bounded by $1$ we obtain that
    \begin{align*}
        \lim_{n \to \infty} &\E_{\Q} \left[ f'(F^{\theta}) \partial_{\theta} F^{\theta} \Phi^{\theta} \phi(n \rho_{A_n^{\theta}})\right] + \E_{\Q} \left[f(F^{\theta}) \partial_{\theta} \Phi^{\theta} \phi(n\rho_{A_n^{\theta}}) \right]\\
        = &\E_{\Q} \left[ f'(F^{\theta}) \partial_{\theta} F^{\theta} \Phi^{\theta} \right] + \E_{\Q} \left[f(F^{\theta}) \partial_{\theta} \Phi^{\theta}  \right] \\
        = & \partial_{\theta} \E[f(F^{\theta}) \Phi^{\theta}].
    \end{align*}
    Moreover, since $\phi(n\rho_{A_n^{\theta}}) \to 1$, almost surely, the limit candidate (if it exists) of the sequence
    \[
    \E_{\Q} \left[ f(F^{\theta}) \delta\left( \frac{\partial_{\theta} F^{\theta}  \Phi^{\theta} \phi(n \rho_{A_n^{\theta}})h}{\int_0^T D_tF^{\theta} h_t \mathrm{d} t}\right)\right]
    \]
    is
    \[
    \E_{\Q} \left[ f(F^{\theta}) \delta\left( \frac{\partial_{\theta} F^{\theta}  \Phi^{\theta} h}{\int_0^T D_tF^{\theta} h_t \mathrm{d} t}\right)\right].
    \]
    Now, using that
    \[
    \E_{\Q} \left[ f(F^{\theta}) \delta\left( \frac{\partial_{\theta} F^{\theta}  \Phi^{\theta} \phi(n \rho_{A_n^{\theta}})h}{\int_0^T D_tF^{\theta} h_t \mathrm{d} t}\right)\right] = \E_{\Q} \left[ f'(F^{\theta}) \partial_{\theta} F^{\theta} \Phi^{\theta} \phi(n \rho_{A_n^{\theta}})\right]
    \]
    and recalling the fact that the right-hand-side sequence is convergent we find that the left-hand-side sequence must be convergent. By uniqueness of the limit of sequences, we conclude that
    \[
    \partial_{\theta} \E_{\Q} \left[ f(F^{\theta}) \Phi^{\theta}  \right] =  
    \E_{\Q} \left[ f(F^{\theta}) \delta\left( \frac{\partial_{\theta} F^{\theta}  \Phi^{\theta} h}{\int_0^T D_tF^{\theta} h_t \mathrm{d} t}\right)\right] + \E_{\Q} \left[f(F^{\theta}) \partial_{\theta} \Phi^{\theta}  \right]
    \]
    as desired.
\end{proof}

\begin{remark}
    In general, finding bounds for the Greeks which are independent of the model is not an easy problem. However, for the $\Delta$ case one can easily derive the estimate $0 \leq \Delta \leq 1$ for Call options and $-1 \leq \Delta \leq 0$ for Put options. Indeed, assume the payoff is $f(S_T) = (S_T-K)^+$. Notice that we can write $S_T$ as
    \[
    S_T = S_0\exp\left( rT - \int_0^T \frac{\sigma(V_s)^2}{2} \d s + \int_0^T \sigma(V_s) \d W_s \right) =: S_0 \mathcal{E}_T.
    \]
    Since $\mathcal{E}_T$ does not depend on $S_0$, clearly $\partial_{S_0} S_T = \mathcal{E}_T = \frac{S_T}{S_0}$ a.s. Now, by definition, the Delta can be computed as
    $$
    \Delta = \partial_{S_0}\E_{\Q}\left[e^{-rT}(S_T - K)^{+} \right]
    $$
    where
    $$
    (S_T - K)^{+} = 
    \begin{cases}
    S_T - K & S_0 \geq \frac{K}{\mathcal{E}_T}, \\
    0 & S_0 < \frac{K}{\mathcal{E}_T}.
    \end{cases}
    $$
    Hence,
    $$
    \partial_{S_0}(S_T-K)^{+} = \partial_{S_0}S_T \1_{\{S_0 \geq \frac{K}{\mathcal{E}_T}\}} = \frac{S_T}{S_0}\1_{\{S_0 \geq \frac{K}{\mathcal{E}_T}\}} =  \frac{S_T}{S_0}\1_{\{S_T \geq K\}}.
    $$
    Plugging this into the definition of the Delta we have
    $$
    \Delta = \frac{1}{S_0} \E_{\Q}\left[ e^{-rT} S_T \1_{\{S_T \geq K\}}\right]
    $$
    Let's proceed with the estimates. On the one hand, since $S_T$ and $S_0$ are positive, we clearly have
    $$
    \Delta = \frac{1}{S_0} \E\left[ e^{-rT} S_T \1_{\{S_T \geq K\}}\right] \geq 0.
    $$
    On the other hand, since $\1_{\{S_T \geq K\}} \leq 1$, we have
    $$
    \Delta = \frac{1}{S_0} \E_{\Q}\left[ e^{-rT} S_T \1_{\{S_T \geq K\}}\right] \leq \frac{1}{S_0} \E_{\Q}\left[ e^{-rT} S_T \right] = \frac{S_0}{S_0} = 1
    $$
    because discounted stock prices are martingales. As a conclusion,
    $$
    0 \leq \Delta \leq 1
    $$
    for a Call option. Doing the same procedure for $(K-S_T)^{+}$ we find that, for a Put option,
    $$
    -1 \leq \Delta \leq 0.
    $$
\end{remark}
Until now we did not have to specify a particular volatility process. From now on we will consider several particular cases of rough Volterra volatility processes.

\subsection{Rough Volterra SV models}\label{ssec:rVolterra_SV}

As before, we assume our price process $S=\{S_{t}, t\in[0,T]\}$ is a strictly positive process under a market chosen risk neutral probability measure $\Q $ that follows the model

\begin{eqnarray}\label{e:S}
S_{t}= S_0+\int_0^t  r S_{u} \d u + \int_0^t \sigma(V_u) S_{u} \d W_u,
\end{eqnarray}
which is the integral form of \eqref{e:dS_t}.

Let $Y=\{Y_t, t\geq 0\}$ be the Gaussian Volterra process correlated with the price process $S$ in the following sense. 
\begin{align}
Y_t &= \int_0^t K_H (t,s)\,\d Z_s, \\
Z_s &= \rho\, W_s + \sqrt{1-\rho^2}\, \widetilde{W}_s, \qquad \rho\in [-1,1],
\end{align}
where $\widetilde{W}$ is also a Wiener process (and hence is the process $Z$). 

In the following, we will distinguish two qualitatively different cases. Either the process $V$ is given explicitly as some known functional of the Gaussian process $Y$, i.e., 
\begin{equation}\label{e:V_t:func}
V_t = f(t,Y_t), 
\end{equation}
where $f$ is a function differentiable in the second variable, or $V$ can be given as a solution to the stochastic differential equation
\begin{align}
V_t &= V_0 + \int_0^t u(V_s) \d s + \int_0^t v(V_s) K_H(t,s) \d Z_s, \label{e:V_t}
\end{align}
where functions $u,v$ are assumed to be in $C^2(\mathbb R)$. Moreover we will assume that $u(V_t)$ and $v(V_t)$ are square-integrable processes. Note that 
\begin{enumerate}
\item the case $K_H (t,s)=\1_{[0,t]}(s)$ reduces this model in both cases to the classical non-fractional SV model, see also Appendix \ref{sec:appendix1},

\item $\rho$ is the correlation between processes $V$ and $S$.
\end{enumerate}

It is worth to mention that some of the Volterra processes have rather complicated formulas for functions $u$ and $v$ in the representation \eqref{e:V_t}, but have on the other hand some nice closed forms \eqref{e:V_t:func}. Calculation of Greeks in these models can be simplified and this is the major reason to treat these two cases separately. In particular, we will consider two classes of rough Volterra SV models:

\subsubsection*{Rough Volterra SV models - Examples}

\begin{enumerate}[label={(\roman*)}]
\item In the $\alpha$RFSV model, first introduced by \cite{MerinoPospisilSobotkaSottinenVives21ijtaf}, $\sigma(x)={x}$ and the variance process $V_t$ is given explicitly by \eqref{e:V_t:func} with $$f(t,x)=V_0\cdot\exp\left\{ \xi\cdot x - \frac12\alpha\xi^2 r(t)\right\}.$$ In particular
\begin{itemize}
\item for $\alpha=1$ we get the rough Bergomi model, which in the $r=0$ case coincide with the rough SABR($\beta=1$) model, and 
\item for $\alpha=0$ we get the non-stationary RFSV model.
\end{itemize}

\item In a new mixed $\alpha$RFSV model, a two factor rough model, two fractional volatility factors are considered, one with $H<\frac{1}{2}$ and another with $H'>\frac{1}{2}$. This model is a generalisation of the $\alpha$RFSV model, see \cite{MerinoPospisilSobotkaSottinenVives21ijtaf} and the mixed fractional Bergomi model, see \cite{AlosLeon2021}.
\item In the rough Stein and Stein model, $\sigma(x)=x$, $u(x)=\kappa (\theta-x)$ and $v(x)=\nu$. Here $V$ is the volatility process.
\item Note that the case $u\equiv v\equiv 0$ is the classical Black-Scholes model.
\end{enumerate}

\section{Greeks formulas in rough Volterra SV models}\label{sec:results1}

In this section we calculate Greeks formulas for the $\alpha$RFSV model (Section \ref{ssec:aRFSV}), for the mixed $\alpha$RFSV model (Section \ref{ssec:maRFSV}) and for the Stein and Stein model (Section \ref{ssec:rSS}). 

Notice that the formulas given in Theorem \ref{t:greeks_general} depend on the process $h$ and the Skorohod integral $\delta$. In the proof of Theorem \ref{t:greeks_general}, the parameter $\varepsilon > 0$ did not play an important role for the proof, but will play instead an important role in the numerics. Indeed, since we can choose $h$ such that $||\1_{[0,T]}(t) - h_t||_{L^2(\Omega\times [0,T])} \leq \varepsilon$ for $\varepsilon > 0$ arbitrarily small, we will simulate the formulas of Theorem \ref{t:greeks_general} for $h_t = \1_{[0,T]}(t)$. Using the Theorem \ref{t:integral_of_product}, we can derive approximate formulas for the Greeks.

\begin{proposition} \label{p: approximate Greek formulas}
    Let $\mathcal{G}_T := \int_0^T G(u,T) \d u$. The approximate formulas for the Greeks are:
    \begin{align*}
    \Delta = & \frac{\e^{-rT}}{S_0} \E_{\Q}\left[ f(S_T) \left( \frac{W_T}{\mathcal{G}_T} + \frac{\int_0^T D_s\mathcal{G}_T \d s}{\mathcal{G}_T^2} \right)\right], \\
    \Gamma = & -\frac{\e^{-rT}}{S_0^2}\E_{\Q}\left[ f(S_T) \left( \frac{W_T}{\mathcal{G}_T} + \frac{\int_0^T D_s\mathcal{G}_T \d s}{ \mathcal{G}_T ^2}\right)\right] \\
    &+\frac{\e^{-rT}}{S_0^2}\E_{\Q} \left[ f(S_T) \left(\frac{W_T}{\mathcal{G}_T} + \frac{\int_0^T D_s\mathcal{G}_T \d s}{\mathcal{G}_T^2} \right)^2\right] \\
    &-\frac{\e^{-rT}}{S_0^2}\E_{\Q} \left[ f(S_T) \left( \frac{T}{\mathcal{G}_T} - \frac{W_T \int_0^T D_t \mathcal{G}_T \d t}{\mathcal{G}_T^2} + \frac{\int_0^T \int_0^T D_t D_s \mathcal{G}_T \d s \d t}{\mathcal{G}_T} - \frac{2\left(\int_0^T D_s\mathcal{G}_T \d s \right)^2}{\mathcal{G}_T^3} \right) \right], \\
    \varrho = &rT\e^{-rT}\E_{\Q}\left[ f(S_T) \left(\frac{W_t}{\mathcal{G}_T} + \frac{\int_0^T D_s \mathcal{G}_T \d s}{\mathcal{G}_T^2} \right)\right] - \e^{-rT}\E_{\Q}[f(S_T)], \\
    \mathcal{V} = & \e^{-rT} \E_{\Q}\left[ f(S_t) \left( \frac{W_T \int_0^T b_{V_0}(V_s) \d s}{\mathcal{G}_T} - \frac{\int_0^T \int_0^T D_sb_{V_0}(V_u) \d u \d s}{\mathcal{G}_T} + \frac{\int_0^T D_s \mathcal{G}_T \left(\int_0^T b_{V_0}(V_u) \d u \right) \d s}{\mathcal{G}_T^2}\right)\right] \\
    & + \e^{-rT}\E_{\Q}\left[ f(S_T)\left( \int_0^T a_{V_0}(V_s) \d W_s \left(\frac{W_T}{\mathcal{G}_T}-\frac{\int_0^T D_s \mathcal{G}_T}{\mathcal{G}_T^2} \right) - \frac{\int_0^T a_{V_0}(V_t) + \int_0^T D_t a_{V_0}(V_s) \d W_s \d t}{\mathcal{G}_T}  \right) \right], \\
    \mathcal{H} = & \e^{-rT} \E_{\Q}\left[ f(S_t) \left( \frac{W_T \int_0^T b_{H}(V_s) \d s}{\mathcal{G}_T} - \frac{\int_0^T \int_0^T D_sb_{H}(V_u) \d u \d s}{\mathcal{G}_T} + \frac{\int_0^T D_s \mathcal{G}_T \left(\int_0^T b_{H}(V_u) \d u \right) \d s}{\mathcal{G}_T^2}\right)\right] \\
    & + \e^{-rT}\E_{\Q}\left[ f(S_T)\left( \int_0^T a_{H}(V_s) \d W_s \left(\frac{W_T}{\mathcal{G}_T}-\frac{\int_0^T D_s \mathcal{G}_T}{\mathcal{G}_T^2} \right) - \frac{\int_0^T a_{H}(V_t) + \int_0^T D_t a_{H}(V_s) \d W_s \d t}{\mathcal{G}_T}  \right) \right],
\end{align*}
where
\begin{align*}
    D_t a_{\theta}(V_s) &= D_t \partial_{\theta} (V_s)\\
    &= \partial_{\theta} \sigma'(V_s) D_t V_s\\
    &= \sigma''(V_s) D_t V_s \partial_{\theta} V_s + \sigma'(V_s) \partial_{\theta} D_tV_s
\intertext{and}
    D_t b_{\theta}(V_s) &= D_t \partial_{\theta} \left(r - \frac{1}{2}\sigma^2(V_s)\right)\\
    &= \partial_{\theta} D_t \frac{-\sigma^2(V_s)}{2} \\
    &= - \partial_{\theta}\sigma(V_s) \sigma'(V_s) D_tV_s \\
    &= -\left( \sigma'(V_s)^2  + \sigma(V_s) \sigma''(V_s) \right)D_tV_s \partial_{\theta} V_s - \sigma(V_s)\sigma'(V_s) \partial_{\theta}(D_tV_s).
\end{align*}
\end{proposition}

\subsection{\texorpdfstring{$\alpha$}{a}RFSV model}\label{ssec:aRFSV}

In the $\alpha$RFSV model firstly introduced by \cite{MerinoPospisilSobotkaSottinenVives21ijtaf} we assume that $\sigma(x) = x$ and that the volatility process is
\begin{equation}\label{e:aRFSV}
V_t = V_0\exp\left\{\xi Y_t - \frac12\alpha\xi^2 r(t) \right\}, \quad t\geq 0,
\end{equation}
where $V_0>0$, $\xi>0$ and $\alpha\in[0,1]$ are model parameters together with $H<1/2$. Recall that $r(t)=t^{2H}$ for the case $$K_H(t,s)=\sqrt{2H}(t-s)^{H-\frac{1}{2}}.$$

For $\alpha = 0$ this model becomes the non-stationary {RFSV model} \citep{Gatheral18} and for $\alpha = 1$ we get the rough Bergomi model \citep{Bayer16}, which is also, in the null interest rate case, the special case of the SABR($\beta=1)$ model. Values of $\alpha$ between zero and one gives the model one more degree of freedom in the sense of stationarity and it is not rare that calibrations to real market data give us these values \citep{MatasPospisil23aofi}.

To calculate the considered Greeks, we need to calculate $D_s V_r$,  $D_t D_s V_r$, $G(t,T)$, $D_s G(t,T)$, $D_tD_sG(t,T)$ and the derivatives of $V_t$ and $DV_t$ with respect to $V_0$ and $H$.

Using the fact that $D_sY_t=\rho K_H(t,s)$ as a consequence of Proposition~\ref{deriv-int}, we have 

$$D_s V_r=\rho \xi V_r K_H(r,s)\1_{[0,r]}(s)$$
and 
$$D_t D_s V_r=\rho^2 \xi^2 K_H(r,s) K_H(r,t) V_r \1_{[0,r]}(s\vee t). $$

Moreover, from \eqref{e:aRFSV} we have that
\[
a_{V_0}(V_t) = \partial_{V_0} V_t = \frac{V_t}{V_0}, \quad b_{V_0}(V_s) =\partial_{V_0}\left( r - \frac{V_s^2}{2}\right) =  \frac{-V_s^2}{V_0}.
\]
Since $\sigma(x) = x$, the formulas needed to compute $\mathcal{V}$ are quite simple. Indeed,
\[
D_t a_{V_0}(V_s) = \partial_{V_0} D_t V_s = \frac{1}{V_0} \rho \xi V_s K_H(s,t) \1_{[0,s]}(t).
\]
Hence,
\[
\int_0^T a_{V_0}(V_t) + \int_0^T D_t a_{V_0}(V_s) \d W_s \d t = \frac{1}{V_0}\left(\int_0^T V_t \d t + \int_0^T V_s \kappa_s \d W_s \right).
\]
In the same way,

\begin{eqnarray*}
D_tb_{V_0}(V_s) &=&-D_t V_s \cdot \frac{V_s}{V_0} - V_s \partial_{V_0} D_t V_s\\ 
&=&-\frac{2}{V_0}D_tV_s\\
&=&\frac{2}{V_0}\rho \xi V_s K_H(s,t) \1_{[0,s]}(t).
\end{eqnarray*}

Hence,
\[
\int_0^T b_{V_0}(V_u) \d u = -\frac{1}{V_0}\int_0^T V_u^2 \d u, \quad \int_0^T \int_0^T D_s b_{V_0}(V_u) \d u \d s = -\frac{2}{V_0}\rho \xi \int_0^T V_u \kappa_u \d u.
\]
Regarding the derivatives with respect to $H$ we have that
\[
a_H(V_t) = \partial_{H} V_t = V_t \cdot \left( \xi Y'_t - \alpha \xi^2 t^{2H} \ln(t) \right)
\]
where
\[
Y'_t = \int_0^t \partial_H K_H(t,s) \d W_s.
\]
Hence,
\begin{align*}
    D_s a_H(V_u) = \partial_{H} D_sV_u = &\rho \xi\partial_H V_r K_H(u,s) \1_{[0,u]}(s) + \rho \xi V_u \partial_H K_H(u,s) \1_{[0,u]}(s) \\
    = &\left[ \rho \xi V_u \left(\xi Y_u' - \alpha \xi^2 u^{2H} \ln(u)\right)K_H(u,s) + \rho \xi V_u \partial_H K_H(u,s)\right]\1_{[0,u]}(s).
\end{align*}
This implies that
\begin{align*}
    &\int_0^T a_{H}(V_t) + \int_0^T D_ta_H(V_s) \d W_s \d t \\
    = &\int_0^T V_t \cdot \left( \xi Y'_t - \alpha \xi^2  t^{2H}\ln(t)\right) \d t + \rho \xi \int_0^T \left[ V_s(\xi Y_s' - \alpha \xi^2 s^{2H}\ln(s))\kappa_s + V_s \kappa'_s\right] \d W_s.
\end{align*}
Similarly, for $b_H$ we have
\[
b_H(V_t) = -\partial_H \frac{V_t^2}{2} = -V_t \partial_H V_t = - V_t^2(\xi Y_t' - \alpha \xi^2 t^{2H} \ln(t))
\]
and
\begin{align*}
\int_0^T \int_0^T D_sb_H(V_u) \d u \d s= -&\int_0^T \int_0^T \partial_H V_u \cdot D_sV_u - V_u \cdot \partial_H(D_sV_u) \d u \d s \\
= -  &\rho \xi\int_0^T  V_u^2 \left[2(\xi Y_u' - \alpha \xi^2 u^{2H} \ln(u)) \kappa_u + \kappa'_u \right] \d u
\end{align*}
From the definition of $G$ \eqref{gformula} and taking $\sigma(x)=x$ we have 

\begin{align*}
G(t,T)
&= V_t +\int_t^T D_t V_r \d W_r - \int_t^T V_r D_t V_r \d r\\
&= V_t+\rho\xi\int_0^t V_r K_H(r,t) \d W_r - \rho \xi \int_t^T V_r^2  K_H(r,t)\d r.
\end{align*}

Finally, using Lemma~\ref{l:DG},

\begin{align*}
D_s G(t,T)
&= D_s V_t + \rho\xi V_s K_H(s,t)\1_{[t,T]}(s)+\int_{t\vee s}^T  \rho\xi K_H(r,t)D_sV_r \d W_r-2\int_{t\vee s}^T \rho \xi K_H(r,t) V_rD_sV_r \d r\\
&= \rho \xi V_{t\vee s} K(t\vee s,s\wedge t)+\rho^2 \xi^2\int_{t \vee s }^T  K(r,s) K(r,t) V_r \d W_r- 2\rho^2 \xi^2\int_{t\vee s}^T  K(r,s) K(r,t) V^2_r \d r.
\end{align*}

Using these formulas we have 

\begin{align}
\int_0^T G(t,T) \d t
&= \int_0^T V_t \d t + \rho\xi \int_0^T \int_t^T V_r K_H(r,t) \d W_r \d t - \rho\xi \int_t^T V_r^2 K_H(r,t) \d r \d t \notag\\
&= \int_0^T V_t \d t + \rho\xi \int_0^T V_r \int_0^r K_H(r,t) \d t \d W_r - \rho\xi \int_0^T V_r^2 \int_0^r K_H(r,t) \d t \d r \notag\\
&= \int_0^T V_t \d t + \rho\xi \left( \int_0^T V_r \kappa_r \d W_r - \int_0^T V_r^2 \kappa_r \d r\right). \label{e:aRFSF_G}
\end{align}
Differentiating this last expression, we obtain
\begin{align}
\int_0^T \int_0^T D_s G(t,T) \d t \d s
&= \int_0^T \int_0^T\rho \xi V_{t\vee s} K(t\vee s,s\wedge t)\d s \d t \notag\\
&\quad + \int_0^T \int_0^T \rho^2 \xi^2\int_{t \vee s }^T  K(r,s) K(r,t) V_r \d W_r \d s \d t \notag\\
&\quad - \int_0^T \int_0^T 2\rho^2 \xi^2\int_{t\vee s}^T  K(r,s) K(r,t) V^2_r \d r\d s \d t \notag\\
&= 2\rho\xi \int_0^T \int_0^t V_t K_H(t,s) \d s \d t \notag\\
&\quad+ \rho^2\xi^2 \int_0^T V_r \left( \int_0^r \int_0^r K_H(r,t)K_H(r,s) \d s \d t\right) \d W_r \notag\\
&\quad- 2\rho^2\xi^2 \int_0^T V_r^2 \left(\int_0^r \int_0^r K_H(r,t) K_H(r,s) \d s \d t\right) \d r \notag\\
&= 2\rho\xi \int_0^T V_t \kappa_t \d t + \rho^2\xi^2 \left(\int_0^T V_r \kappa_r^2 \d W_r 
- 2\int_0^T V_r^2 \kappa_r^2 \d r\right) \label{e:aRFSV_DG}
\end{align}
and finally,
\begin{align*}
    \int_0^T \int_0^T \int_0^T D_tD_sG(u,t) \d u \d s \d t = & 2\rho^2 \xi^2 \int_0^T \kappa_u^2 V_u \d u \\
    + &\rho^2 \xi^2 \int_0^T \kappa_u^2 V_u \d u + \rho^3 \xi^3 \int_0^T \kappa_u^3 V_u \d W_u \\
    - & 4 \rho^3 \xi^3 \int_0^T \kappa_u^3 V_u^2  \d u \\
    = & 3\rho^2 \xi^2 \int_0^T \kappa_u^2 V_u \d u - \rho^3 \xi^3 \left( \int_0^T \kappa_u^3 V_u \d W_u - 4 \int_0^T \kappa_u^3 V_u^2\right).
\end{align*}
In conclusion, the tools needed to compute the Greeks for the $\alpha$RFSV model are the following:
\begin{itemize}
    \item $\mathcal{G}_T = \int_0^T V_t \d t + \rho\xi \left( \int_0^T V_r \kappa_r \d W_r - \int_0^T V_r^2 \kappa_r \d r\right)$,
    \item $\int_0^T D_s \mathcal{G}_T \d s = 2\rho\xi \int_0^T V_t \kappa_t \d t + \rho^2\xi^2 \left(\int_0^T V_r \kappa_r^2 \d W_r 
- 2\int_0^T V_r^2 \kappa_r^2 \d r\right)$,
    \item $ \int_0^T \int_0^T D_t D_s \mathcal{G}_T \d s \d t = 3\rho^2 \xi^2 \int_0^T \kappa_u^2 V_u \d u - \rho^3 \xi^3 \left( \int_0^T \kappa_u^3 V_u \d W_u - 4 \int_0^T \kappa_u^3 V_u^2\right)$,
    \item $a_{V_0}(V_t) =  \frac{V_t}{V_0}$,
    \item $b_{V_0}(V_t) = - \frac{V_t^2}{2}$,
    \item $\int_0^T a_{V_0}(V_t) + \int_0^T D_t a_{V_0}(V_s) \d W_s \d t = \frac{1}{V_0}\left(\int_0^T V_t \d t + \int_0^T V_s \kappa_s \d W_s \right)$,
    \item $\int_0^T \int_0^T D_sb_{V_0}(V_u) \d u \d s = \frac{-2}{V_0}\rho \xi \int_0^T V_u \kappa_u \d u$,
    \item $a_H(V_t) = V_t \cdot \left( \xi Y'_t - \alpha \xi^2  t^{2H}\ln(t)\right)$, where $Y'_t = \int_0^t \partial_H K_H(t,s) \d W_s$, 
    \item $b_H(V_t) = - V_t^2\cdot \left( \xi Y'_t - \alpha \xi^2  t^{2H}\ln(t)\right)$,
    \item %
    \begin{align*}
        &\int_0^T a_{H}(V_t) + \int_0^T D_ta_H(V_s) \d W_s \d t \\
    &= \int_0^T V_t \cdot \left( \xi Y'_t - \alpha \xi^2  t^{2H}\ln(t)\right) \d t + \rho \xi \int_0^T \left[ V_s(\xi Y_s' - \alpha \xi^2 s^{2H}\ln(s))\kappa_s + V_s \kappa'_s\right] \d W_s,
    \end{align*}
    \item $\int_0^T \int_0^T D_sb_H(V_u) \d u \d s = -  \rho \xi\int_0^T  V_u^2 \left[2(\xi Y_u' - \alpha \xi^2 u^{2H} \ln(u)) \kappa_u + \kappa'_u \right] \d u$.
\end{itemize}
These are all ingredients necessary to plug into the Greeks formulas in Proposition \ref{p: approximate Greek formulas}.

\begin{remark}
Greeks formulas for the rough Bergomi model can be easily obtained from the above formulas by taking $\alpha=1$, and similarly formulas for the non-stationary RFSV model or rough SABR($\beta=1$) model by taking $\alpha=0$.
\end{remark}

In Figure \ref{f:aRFSV_Delta} we can see a convergence of the Delta for the $\alpha$RFSV model with $H=0.15$, $\alpha=1$ (rBergomi), $V_0=0.354$, $\xi=0.212$ and $\rho=-0.756$. Numerically, the kernel \eqref{e:K_H_eps} is considered with $\varepsilon=10^{-6}$. Market values are $S_0=100$ (dollars) and $r=0.05$, option parameters are strike price $K=100$ (dollars) and maturity $T=1$ (year). Monte-Carlo (MC) simulation point estimate together with the 99\% confidence interval (CI) is plotted as a dependence on the number of simulations (NS). In particular $\Delta_{\text{call}} = 0.3399 \pm 0.0051$ and $\Delta_{\text{put}} = -0.9220 \pm 0.0156$. Although we have no analytical $\alpha$RFSV formula for Delta, we plot the horizontal dashed line which is the numerical mean at the end (for the largest number of simulations).

\begin{figure}[ht]
\includegraphics[width=0.49\textwidth,trim=0mm 0mm 20mm 0mm,clip]{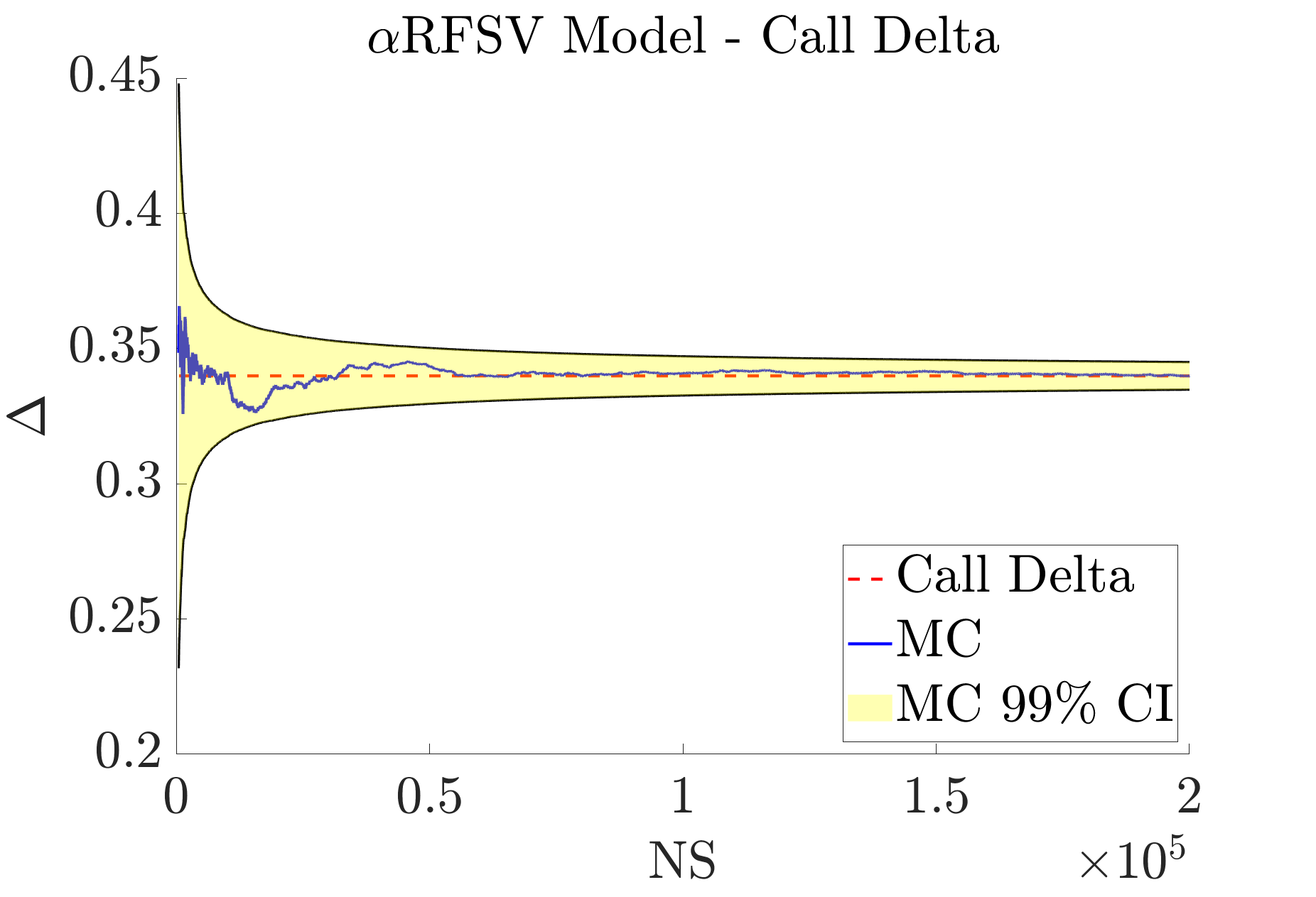}
\includegraphics[width=0.49\textwidth,trim=0mm 0mm 20mm 0mm,clip]{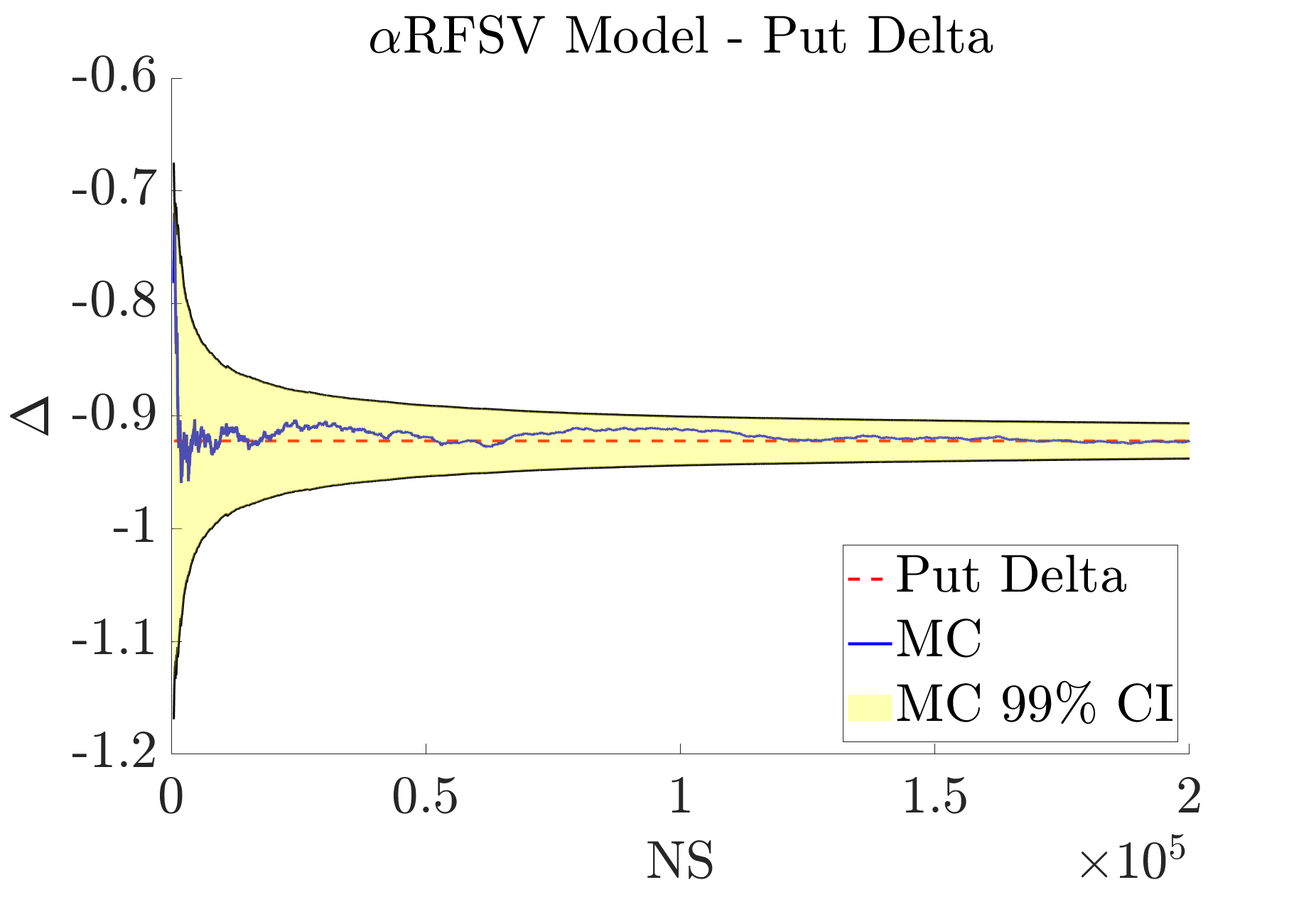}
\caption{Convergence of the $\alpha$RFSV Delta in the case $H=0.15$.}\label{f:aRFSV_Delta}
\end{figure}

In Figure \ref{f:aRFSV_Delta_on_H} we depict the dependence of the $\alpha$RFSV Delta on $H$, both for the call and put option case, with remaining parameters left unchanged. 

\begin{figure}[ht]
\includegraphics[width=0.49\textwidth,trim=0mm 0mm 20mm 0mm,clip]{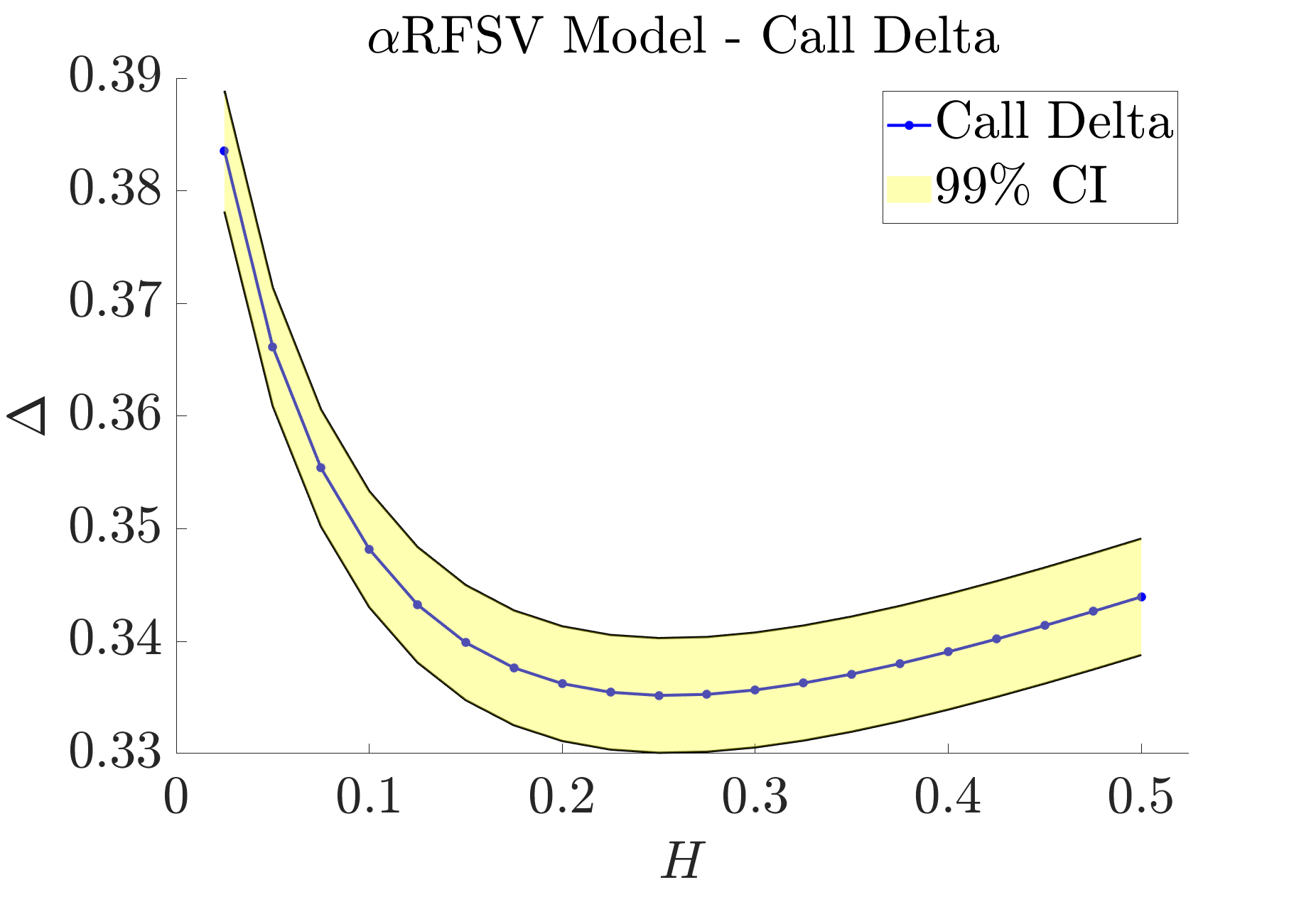}
\includegraphics[width=0.49\textwidth,trim=0mm 0mm 20mm 0mm,clip]{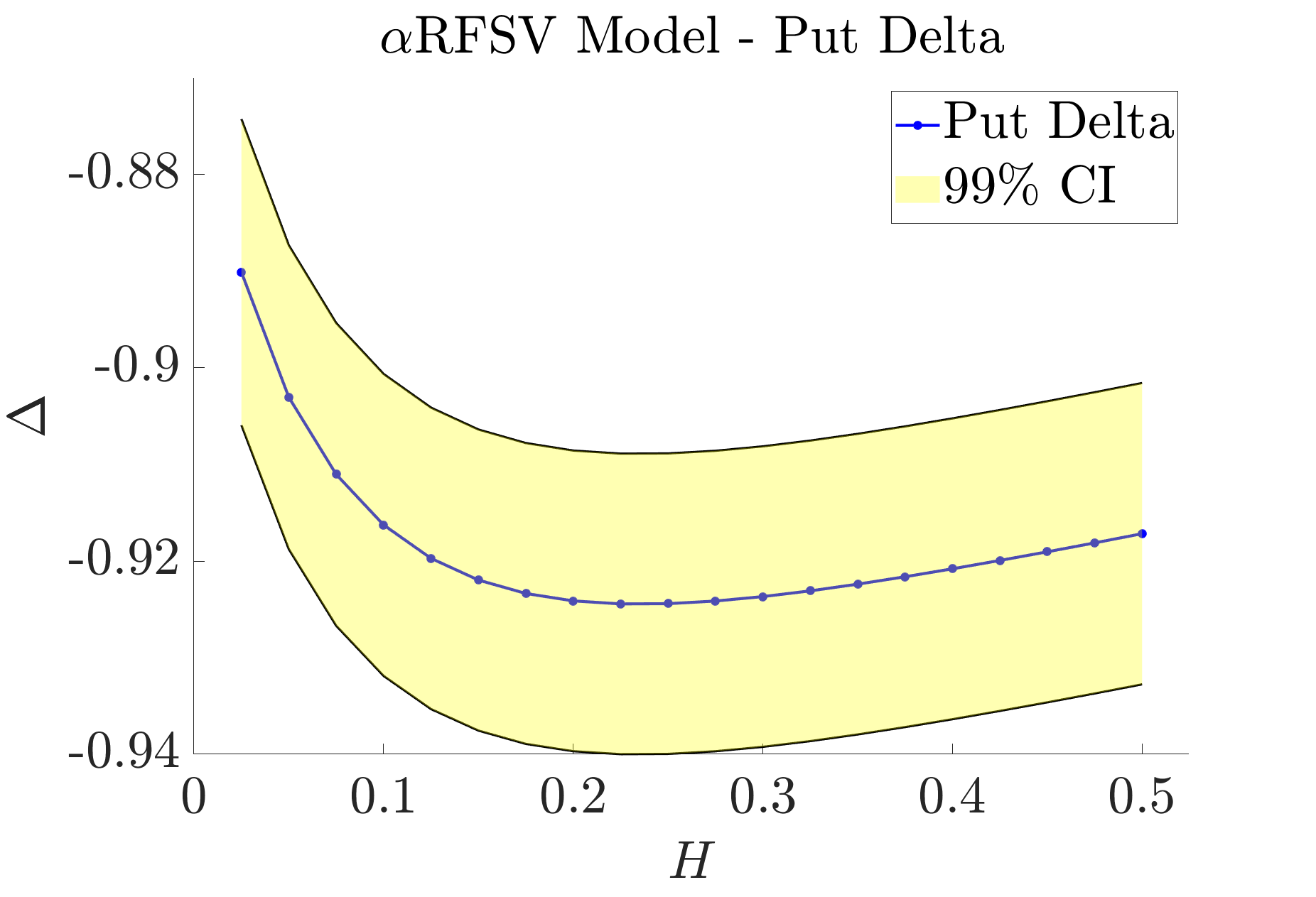}
\caption{Dependence of the $\alpha$RFSV Delta on the Hurst parameter $H$.}\label{f:aRFSV_Delta_on_H}
\end{figure}

\subsection{Mixed \texorpdfstring{$\alpha$}{a}RFSV model}\label{ssec:maRFSV}

In this section we follow \cite{AlosLeon2021}, who considered a mixed fractional Bergomi model in which the volatility process is an arithmetic average of two fractional processes, each with a different Hurst parameter. In particular we consider 

\begin{equation}\label{mixedV_t}
V_t=\frac{1}{2}[V_t^{H}+V_t^{H^{'}}],
\end{equation}
where $V_t^H$ and $V_t^{H^{'}}$ are $\alpha$RFSV processes defined by \eqref{e:aRFSV} with the Hurst parameters $H<\frac{1}{2}$ and $H^{'}>\frac{1}{2}$. In this setting, volatility processes $V_t^{H}$ and $V_t^{H^{'}}$ represents the short-memory and long-memory factors respectively.

Calculations that we performed in Section \ref{ssec:aRFSV} have to be changed in the following way:

\begin{equation}
D_s V_r = \frac{\rho}{2}[\xi^H V_r^H K_{H}(r,s)+\xi^{H^{'}}V_r^{H{'}}K_{H^{'}}(r,s)]{1\!\!1}_{[0,r]}(s)
\end{equation}
and

\begin{equation}
D_t D_s V_r= \frac{\rho}{2}
[(\xi^{H})^{2} K_H(r,s) K_H(r,t) V^H_r+(\xi^{H^{'}})^{2} K_{H^{'}}(r,s) K_{H^{'}}(r,t) V^{H^{'}}_r]{1\!\!1}_{[0,r]}(s\vee t).
\end{equation}
In order to get all the ingredients needed to compute the Greeks for this model, one should repeat the computations for the $\alpha$RFSV model, add them and divide them by 2. Examples of how these computations work are, for example 
\begin{align*}
G(t,T) &= V_t +\int_t^T D_t V_r dW_r-\int_t^T V_r D_t V_r\d r\\
&= \frac{1}{2}[V_t^{H}+V_t^{H^{'}}] \\
&\quad+\frac{\rho}{2}\left(\xi^H \int_t^T  V_r^H K_{H}(r,t) \d W_r +\xi_{H^{'}} \int_t^T V_r^{H^{'}} K_{H^{'}}(r,t) \d W_r\right)\\
&\quad -\frac{\rho}{2}\left(\xi^{H}\int_t^T (V_r^H)^2 K_{H}(r,t) \d r+\xi^{{H}^{'}}\int_t^T (V_r^{{H}^{'}})^2 K_{H^{'}}(r,t) \d r\right)
\end{align*}
and

\begin{align*}
\int_0^T G(t,T)\d t &= \frac{1}{2}\int_0^T[V_t^{H}+V_t^{H^{'}}]\d t\\
&\quad+\frac{\rho}{2}\left(\xi^H\int_0^T  \int_t^T  V_r^H K_{H}(r,t) \d W_r \d t+\xi_{H^{'}} \int_0^T \int_t^T V_r^{H^{'}} K_{H^{'}}(r,t) \d W_r \d t\right)\\
&\quad-\frac{\rho}{2}\left(\xi^{H}\int_0^T\int_t^T (V_r^H)^2 K_{H}(r,t) \d rdt+\xi^{{H}^{'}}\int_0^T \int_t^T  (V_r^{{H}^{'}})^2 K_{H^{'}}(r,s) \d r dt \right).\\
\end{align*}

Finally
\begin{align*}
D_s G(t,T) &= D_{s\wedge t} V_{s\vee t} + \int_{s\vee t}^T D_s D_t V_r \d W_r- \int_{s\vee t}^T D_sV_r \cdot D_t V_r\d r-\int_{s\vee t}^T V_r\cdot D_sD_t V_r\d r \\
&=\frac{\rho}{2}\xi^H V_{s\vee t}^H K_{H}(s\vee t,s\wedge t)+\frac{1}{2}\xi^{H^{'}}V_{s\vee t}^{H{'}}K_{H^{'}}(s\vee t, s\wedge t)\\
&\quad+ \frac{\rho}{2}\int_{s\vee t}^T\left[(\xi^{H})^{2} K_H(r,s) K_H(r,t) V^H_r +(\xi^{H^{'}})^{2} K_{H^{'}}(r,s) K_{H^{'}}(r,t) V^{H^{'}}_r\right]\d W_r\\
&\quad-\frac{\rho}{4}\int_{s\vee t}^T\left[\xi^H V_t^H K_{H}(r,s)+\xi^{H^{'}}V_t^{H{'}}K_{H^{'}}(r,s)\right]\cdot
\left[\xi^H V_t^H K_{H}(r,t) +\xi^{H^{'}}V_r^{H{'}}K_{H^{'}}(r,t)\right]\d r\\
&\quad-\frac{\rho}{4}\int_{s\vee t}^T\left[V_r^{H}+V_r^{H^{'}}\right]\cdot
\left[(\xi^{H})^{2} K_H(r,s) K_H(r,t) V^H_r+(\xi^{H^{'}})^{2} K_{H^{'}}(r,s) K_{H^{'}}(r,t) V^{H^{'}}_r\right]\d r\\
\end{align*}
and 

\begin{align*}
&\int_0^T \int_0^T D_s G(t,T)\d s \d t \\
&=\frac{\rho}{2}\int_0^T \int_0^T \left[\xi^H V_{s\vee t}^H K_{H}(s\vee t, s\wedge t)+\xi^{H^{'}}V_t^{H{'}}K_{H^{'}}(s\vee t, s\wedge t)\right]\d s\d t\\
&\quad+\frac{\rho}{2}\int_0^T \int_0^T\int_{s\vee t}^T \left[(\xi^{H})^{2} K_H(r,s) K_H(r,t) V^H_r +(\xi^{H^{'}})^{2} K_{H^{'}}(r,s) K_{H^{'}}(r,t) V^{H^{'}}_r\right]\d W_r \d s \d t\\
&\quad-\frac{\rho}{4}\int_0^T \int_0^T\int_{s\vee t}^T\left[\xi^H V_t^H K_{H}(r,s)+\xi^{H^{'}}V_t^{H{'}}K_{H^{'}}(r,s)\right]\cdot
\left[\xi^H V_t^H K_{H}(r,t)+\xi^{H^{'}}V_r^{H{'}}K_{H^{'}}(r,t)\right]\d r \d s \d t\\
&\quad-\frac{\rho}{4} \int_0^T \int_0^T\int_{s\vee t}^T\left[V_r^{H}+V_r^{H^{'}}\right]\cdot
\left[(\xi^{H})^{2} K_H(r,s) K_H(r,t) V^H_r +(\xi^{H^{'}})^{2} K_{H^{'}}(r,s) K_{H^{'}}(r,t) V^{H^{'}}_r\right]\d r\d s\d t.
\end{align*}

\subsection{Rough Stein-Stein model}\label{ssec:rSS}
Now we assume the volatility $V$ satisfies the equation

$$V_s = V_0 + \int_0^s \kappa (\theta-V_u)\d u + \nu \int_0^s K_H(s,u)\d Z_u,$$
where $\kappa$, $\theta$ and $\nu$ are positive constants. 

In relation with the existence and uniqueness of solution of this equation we refer the reader to \cite{NualartOuknine2002}. Being the drift term derivable, the Malliavin differentiability of the solution is straightforward. 

Then, for $t\leq s$ we have
$$D_tV_s = \rho\nu K_H(s,t)-\kappa \int_t^s D_tV_u \d u.$$
This is a linear integral equation. Its solution is given by 
$$D_tV_s = \rho\nu K_H(s,t)-\kappa\rho\nu\int_t^s  K_H(u,t) \e^{-\kappa (s-u)}\d u.$$
Being $D_tV_u$ deterministic, $D_sD_tV_u$ is null. Then, applying Lemma \ref{l:DS} we have 
$$G(t,T)=V_t+\int_t^T D_tV_u (\d W_u-V_u \d u)$$
and using  Lemma \ref{l:DG} we have
$$D_sG(t,T)=D_{s\wedge t}V_{s\vee t}-\int_{s\vee t}^T D_tV_u D_sV_u \d u.$$
Then, 
$$\int_0^T G(u,T) \d u = \int_0^T V_u \d u + \rho\nu\int_0^T \left(\int_0^s D_tV_u \d u\right)(\d W_s-V_s \d s),$$
$$\int_0^T\int_0^T D_sG(u,T)\d u\d s = \int_0^T\int_0^T D_{s\wedge u}V_{s\vee u}\d s\d u - \int_0^T \left(\int_0^u D_sV_u \d s\right)^2 \d u.$$
and
$$ \int_0^T \int_0^T \int_0^T D_tD_sG(u,t) \d u \d s \d t = 0.$$
In order to compute the Greeks $\mathcal{V}$ and $\mathcal{H}$, we need to compute the derivatives with respect to $V_0$ and $H$ of the quantities $V_t$, $\frac{-V_t^2}{2}$, $D_tV_s$ and $D_t\left(\frac{-V_t^2}{2} \right)$. On the one hand, notice that
\[
\partial_{V_0}V_t = 1 - \kappa \int_0^t \partial_{V_0}V_u \d u.
\]
Hence,
\[
a_{V_0}(V_t) = \partial_{V_0} V_t = \e^{-\kappa t}, \quad D_sa_{V_0}(V_u) = 0.
\]
In a similar fashion, we have that
\[
b_{V_0}(V_t) = -\e^{-\kappa t}V_t , \quad D_t b_{V_0}(V_s) = -\e^{-\kappa s}\left(\rho\nu K_H(s,t)-\kappa\rho\nu\int_t^s  K_H(u,t) \e^{-\kappa (s-u)}\d u\right).
\]
Concerning the derivatives with respect to $H$ we have
\begin{equation}\label{e: derivative wrt H stein-stein}
\partial_H V_t = -\kappa \int_0^t \partial_H V_u \d u + \int_0^t \partial_H K_H(t,u) \d Z_u.
\end{equation} 
Since
\[
\partial_H K_H(t,s) = \sqrt{2H}(t-s)^{H-1/2} \left(\frac{1}{2H} + \ln(t-s) \right) \in L^2([0,t])
\]
for all $t \in [0,T]$ we have that there exists a unique solution to equation \eqref{e: derivative wrt H stein-stein}. This implies that $a_{H}(V_t)$ is the solution to \eqref{e: derivative wrt H stein-stein} and $D_ta_H(V_s)$ can be expressed as
\[
D_ta_H(V_s) = \partial_H D_tV_s = \rho\nu \partial_H K_H(s,t)-\kappa\rho\nu\int_t^s  \partial_H K_H(u,t) \e^{-\kappa (s-u)}\d u.
\]
Moreover, since $b_H(V_t) = -V_t \partial_H V_t$, we have
\begin{align*}
    D_tb_H(V_s) = &-D_tV_s \cdot \partial_H V_s - V_t \cdot \partial_H(D_tV_s)\\
                = & \partial_H V_s \left(  \rho\nu K_H(s,t)-\kappa\rho\nu\int_t^s  K_H(u,t) \e^{-\kappa (s-u)}\d u\right)\\
                - &V_s\left( \rho\nu \partial_H K_H(s,t)-\kappa\rho\nu\int_t^s  \partial_H K_H(u,t) \e^{-\kappa (s-u)}\d u\right).
\end{align*}
In conclusion, the ingredients needed to compute the approximate Greek formulas for the Stein-Stein model are the following:
\begin{itemize}
    \item $\mathcal{G}_T = \int_0^T V_u \d u + \rho\nu\int_0^T \left(\int_0^s D_tV_u \d u\right)(\d W_s-V_s \d s)$,
    \item $\int_0^T D_s\mathcal{G}_T \d s = \int_0^T\int_0^T D_{s\wedge u}V_{s\vee u}\d s\d u - \int_0^T \left(\int_0^u D_sV_u \d s\right)^2 \d u$,
    \item $\int_0^T \int_0^T D_t D_s \mathcal{G}_T \d s \d t = 0$,
    \item $a_{V_0}(V_t) = \e^{-\kappa t}$,
    \item $b_{V_0}(V_t) = -\e^{-\kappa t} V_t$,
    \item $D_ta_{V_0}(V_s) = 0$,
    \item $D_tb_{V_0}(V_s) = -\e^{-\kappa s}\left(\rho\nu K_H(s,t)-\kappa\rho\nu\int_t^s  K_H(u,t) \e^{-\kappa (s-u)}\d u\right)$.
    \item $a_{H}(V_t) = \partial_H V_t$, where $\partial_H V_t$ is the solution to equation \eqref{e: derivative wrt H stein-stein}.
    \item $b_H(V_t) = -V_t \partial_H V_t$.
    \item $D_t a_H(V_s) = \rho\nu \partial_H K_H(s,t)-\kappa\rho\nu\int_t^s  \partial_H K_H(u,t) \e^{-\kappa (s-u)}\d u$.
    \item 
    \begin{align*}
    D_tb_H(V_s) = &\partial_H V_s \left(  \rho\nu K_H(s,t)-\kappa\rho\nu\int_t^s  K_H(u,t) \e^{-\kappa (s-u)}\d u\right)\\
    &- V_s\left( \rho\nu \partial_H K_H(s,t)-\kappa\rho\nu\int_t^s  \partial_H K_H(u,t) \e^{-\kappa (s-u)}\d u\right).
    \end{align*}
\end{itemize}
\section{Conclusion}\label{sec:conclusion}

We used Malliavin calculus techniques to obtain formulas for computing Greeks under different rough Volterra stochastic volatility models. In particular, we showed that each model is fully characterized by the quantity $G(t,T)$, defined in \eqref{gformula}, and by its Malliavin derivative $D_s G(t,T)$ (see Lemma \ref{l:DG}). Calculating the integrals \eqref{e:iint_DsG} and \eqref{e:int_G} then leads us directly to all the Greeks formulas listed in Theorem \ref{t:greeks_general} and Proposition \ref{p: approximate Greek formulas}. In particular, we derived formulas for rough versions of the Stein-Stein, SABR and Bergomi models.

For some of the stochastic volatility models considered in the present paper, the value of the underlying price at $T$, $S_T$, does not belong to $L^2(\Omega)$. For example for rough Bergomi model, see \cite{Gassiat2019}. This justifies the extension of the integration by parts formula to the $L^1$ case in Theorem \ref{p: IBP extension} and the use of an approximation of the unity in the proof of Theorem \ref{t:greeks_general}.

We derive the formulas for all five primary Greeks $\Delta$, $\Gamma$, $\varrho$, $\mathcal{V}$ and $\mathcal{H}$. We assumed that our price process $S$ followed (\ref{e:S}) and  the volatility or the variance $V$ followed (\ref{e:V_t}). Explicitly, we calculate the Greeks for the $\alpha$RFSV model, that includes rough Bergomi and rough SABR models in some cases, the mixed $\alpha$RFSV model and the rough Stein and Stein model.

As numerical results, we calculated the Delta for $\alpha$RFSV model for European call and put. As we can see in Figure \ref{f:aRFSV_Delta}, Monte Carlo simulations converge satisfactorily. Moreover, for the first time we could present a dependence of Delta on the Hurst parameter $H$ in Figure \ref{f:aRFSV_Delta_on_H}.

Further research directions may involve other rough volatility models such as some rough Heston variant and also a further numerical analysis of all obtained formulas.

\section*{Funding}

The work of Josep Vives was partially supported by Spanish grant PID2020-118339GB-100 (2021-2024). Author Òscar Burés is supported by the predoctoral program AGAUR-FI ajuts (2025 FI-1 00580)

\section*{Acknowledgements}

Computational and storage resources were provided by the e-INFRA CZ project (ID:90254), supported by the Ministry of Education, Youth and Sports of the Czech Republic.

We thank Professor David Nualart for some comments in relation to the reference \cite{NualartOuknine2002}.

\appendix
\section{Greeks formulas in non-fractional Volterra SV models}\label{sec:appendix1}
In this Appendix we present Greeks formulas for the corresponding non-fractional SV models treated in the present paper. Concretely, for the $\alpha$SV model (Section \ref{ssec:aSV}) and for the Stein-Stein model (Section \ref{ssec:SS}).

\subsection{\texorpdfstring{$\alpha$}{a}SV model}\label{ssec:aSV}

Although \cite{MerinoPospisilSobotkaSottinenVives21ijtaf} introduced the $\alpha$RFSV model directly in the (rough) fractional form, it makes a perfect sense to consider also its non-fractional version called $\alpha$SV model, i.e. the case with $H=1/2$ , that also reduces to known models for particular choice of the parameter $\alpha$. In particular, for $\alpha=1$ we get the Bergomi model \citep{Bergomi16}, that in the particular case of $r=0$, coincides with the SABR($\beta=1$) model \citep{Hagan02}. For $\alpha=0$ the model is the non-stationary exponential volatility model presented in \cite{Gatheral18}.

Here we introduce a slightly more general model. In this case, we assume the price process satisfies the equation 

\begin{equation}
\d S_t = r S_t \d t + \sqrt{V_t} S_t \d W_t ,
\end{equation}
and the volatility process is 

$$V_t = V_0\exp\left\{\xi Z_t-\frac12 \alpha\xi^2 t\right\}.$$
Here $V_0$ and $\xi$ are positive constants and $\alpha\in [0,1].$ 

Note that 
$$D_tV_s=\rho\xi V_s {1\!\!1}_{[0,s]}(t). $$
Then, 
$$G(t,T) = \sqrt{V_t}+\frac{\rho\xi}{2}\int_t^T \sqrt{V_s}\d W_s-\frac{\rho\xi}{2}\int_t^T V_s \d s$$
and
$$D_sG(t,T)=\frac{\rho\xi}{2}\sqrt{V_{s\vee t}}+\frac{\rho^2\xi^2}{4}\int_{s\vee t}^T \sqrt{V_u}\d W_u-\frac{\rho^2\xi^2}{2}\int_{s\vee t}^T V_u \d u.$$

\subsection{Stein-Stein model}\label{ssec:SS}
In the Stein and Stein model, see \cite{SteinStein91}, we have the price process

\begin{equation}
\d S_t = r S_t \d t + V_t S_t \d W_t ,
\end{equation}
where $V$ is the volatility process described by the Gaussian Ornstein-Uhlenbeck process 

\begin{equation}
\d V_t = \kappa (\theta -V_t) \d t + \nu \d Z_t.
\end{equation}

Therefore, in this case, $\sigma(x)=x$, $u(x)=\kappa(\theta-x)$ and $v(x)=\nu.$ And first derivatives are $\sigma^{\prime}(x)=1$, $u^{\prime}(x)=-\kappa$ and $v^{\prime}(x)=0.$

Then,
$$D_tV_s=\rho\nu \e^{-\kappa (s-t)}{1\!\!1}_{[0,s]}(t).$$
And this leads to 
\begin{align*}
G(t,T)=& V_t+\int_t^T D_tV_s \d W_s-\int_t^T V_sD_tV_s \d s\\
=& V_t+\rho\nu\int_t^T e^{-\kappa (s-t)} \d W_s-\rho\nu\int_t^T V_s \e^{-\kappa (s-t)}\d s
\end{align*}
and 
\begin{equation}
D_s G(t,T)=\rho\nu \e^{-\kappa (s\vee t-s\wedge t)}-\rho\nu \e^{\kappa (s+t)}\int_{s\vee t}^T \e^{-2\kappa u}\d u.
\end{equation}
From this last equation we also deduce that the second order Malliavin derivative of $G(t,T)$ vanishes.

\begin{switch}{3}
\case{1}{%
\bibliographystyle{references/styles/jp+doi+mr+zbl}
\bibliography{new-greeks,references/references,references/references-books,references/references-own,references/preprints,references/preprints-own}
}
\case{2}{%
\bibliographystyle{jp+doi+mr+zbl}
\bibliography{new-greeks,references-export}
}
\case{3}{%

 }
\default{}
\end{switch}

\end{document}